\documentclass{article}

\usepackage{cmbright}
\usepackage{amssymb,amsmath,amsthm}
\usepackage{mathrsfs}
\usepackage[colorlinks,allcolors=blue]{hyperref}

\def\A{\mathcal A}
\def\B{\mathscr B}
\def\C{\mathbb C}
\def\d{\mathrm{d}}
\def\D{\mathscr D}
\def\dom{\mathcal D}
\def\F{\mathscr F}
\def\G{\mathcal G}
\def\H{\mathcal H}

\def\ltwo{\mathop{\mathrm{L}^2}\nolimits}
\def\ltwoloc{\mathop{\mathrm{L}^2_{\rm loc}}\nolimits}
\def\lone{\mathop{\mathrm{L}^1}\nolimits}
\def\linf{\mathop{\mathrm{L}^\infty}\nolimits}
\def\N{\mathbb N}
\def\R{\mathbb R}
\def\S{\mathbb S}
\def\SS{\mathscr S}
\def\U{\mathscr U}
\def\Z{\mathbb Z}

\def\e{\mathop{\mathrm{e}}\nolimits}

\DeclareMathOperator*{\slim}{s\hspace{0.1pt}-\hspace{0.1pt}lim}
\DeclareMathOperator*{\ulim}{u\hspace{0.1pt}-\hspace{0.1pt}lim}

\newtheorem{Theorem}{Theorem}[section]
\newtheorem{Remark}[Theorem]{Remark}
\newtheorem{Lemma}[Theorem]{Lemma}

\newtheorem{Proposition}[Theorem]{Proposition}
\newtheorem{Definition}[Theorem]{Definition}
\newtheorem{Assumption}[Theorem]{Assumption}
\newtheorem{Example}[Theorem]{Example}

\begin{document}


\title{Quantum time delay for unitary operators: general theory}

\author{D. Sambou\footnote{Supported by the Chilean Fondecyt Grant 3170411.}~~and R.
Tiedra de Aldecoa\footnote{Supported by the Chilean Fondecyt Grant 1170008.}}

\date{\small}
\maketitle
\vspace{-1cm}

\begin{quote}
\emph{
\begin{itemize}
\item[] Facultad de Matem\'aticas, Pontificia Universidad Cat\'olica de Chile,\\
Av. Vicu\~na Mackenna 4860, Santiago, Chile
\item[] \emph{E-mails:} disambou@mat.puc.cl, rtiedra@mat.puc.cl
\end{itemize}
}
\end{quote}


\begin{abstract}
We present a suitable framework for the definition of quantum time delay in terms of
sojourn times for unitary operators in a two-Hilbert spaces setting. We prove that
this time delay defined in terms of sojourn times (time-dependent definition) exists
and coincides with the expectation value of a unitary analogue of the Eisenbud-Wigner
time delay operator (time-independent definition). Our proofs rely on a new summation
formula relating localisation operators to time operators and on various tools from
functional analysis such as Mackey's imprimititvity theorem, Trotter-Kato Formula and
commutator methods for unitary operators. Our approach is general and
model-independent.
\end{abstract}


\smallskip


\tableofcontents

\section{Introduction and main results}\label{section_intro}
\setcounter{equation}{0}

One can find a large literature on the notion of quantum time delay in the setup of
scattering theory for self-adjoint operators (see for instance
\cite{AC87,ACS87,AS06,BGG83,BO79,JSM72,Jen81,Mar81,Nar84,NW16,RT12,Rob94,RW89,Tie09,
Wan88}). But as far as we know, there is no mathematical work on quantum time delay in
the setup of scattering theory for unitary operators. The purpose of this paper is to
fill in this gap by developing a general theory of quantum time delay for unitary
operators in a two-Hilbert spaces setting. Namely, we present a suitable framework
for the definition of time delay in terms of sojourn times in dilated regions for
quantum scattering systems consisting in two unitary operators acting in two (a
priori different) Hilbert spaces. Then, we prove under appropriate conditions the
existence of this time delay and its equality with the expectation value of a unitary
analogue of the Eisenbud-Wigner time delay operator appearing in the self-adjoint
theory. This establishes in a general unitary setup, as was established before in a
general self-adjoint setup \cite{RT12}, the identity between the time-dependent
definition of time delay (with sojourn times) and the time-independent definition of
time delay (with expectation value).

Our framework is the following. Assume that we have a scattering system $(U_0,U,J)$
consisting in a unitary operator $U$ in a Hilbert space $\H$ (a full propagator), a
unitary operator $U_0$ in a Hilbert space $\H_0$ (a free propagator), and a bounded
operator $J:\H_0\to\H$ (an identification operator) such that the wave operators
$$
W_\pm:=\slim_{n\to\pm\infty}U^{-n}J\;\!U_0^nP_{\rm ac}(U_0)
$$
exist and are partial isometries with initial subspaces $\H_0^\pm\subset\H_0$ and
final subspaces $\H_{\rm ac}(U)$. Then, the scattering operator
$$
S:=W_+^*W_-:\H_0^-\to\H_0^+
$$
is a well-defined unitary operator which commutes with $U_0$ and decomposes into a
family of unitary operators $\{S(z)\}_{z\in\sigma(U_0)}$ in the spectral
representation of $U_0$. Assume also that we have a family of mutually commuting
self-adjoint operators $Q:=(Q_1,\ldots,Q_d)$ in $\H_0$ (position operators) satisfying
appropriate commutations relations with respect to $U_0$. Finally, assume that we have
a non-negative even Schwarz function $f\in\SS(\R^d)$ equal to $1$ on a neighbourhood
of $0\in\R^d$ (a localisation function). Then, for any fixed $r>0$ and suitable
$\varphi\in\H_{\rm ac}(U_0)$, we define the sojourn time of the freely evolving state
$U_0^n\varphi$ in the region defined by the localisation operator $f(Q/r)$ as
$$
T_r^0(\varphi)
:=\sum_{n\in\Z}\big\langle U_0^n\varphi,f(Q/r)U_0^n\varphi\big\rangle_{\H_0}.
$$
Now, the
evolution group $\{U^n\}_{n\in\Z}$ acts in $\H$ whereas $f(Q/r)$ acts in $\H_0$.
Therefore, in order to define a sojourn time for the corresponding fully evolving
state $U^nW_-\varphi$, one needs to introduce a family of operators $L_n:\H\to\H_0$ to
inject the operator $f(Q/r)$ in the Hilbert space $\H$ (the operators $L_n$ must
satisfy some natural conditions which in simple cases are verified if $L_n=J^*$ for
all $n\in\Z$). We then define the sojourn time of the fully evolving state
$U^nW_-\varphi$ in the region defined by the localisation operator $f(Q/r)$ as
$$
T_{r,1}(\varphi):=\sum_{n\in\Z}\big\langle L_nU^nW_-\varphi,
f(Q/r)L_nU^nW_-\varphi\big\rangle_{\H_0}.
$$
An additional sojourn time appears naturally in this two-Hilbert spaces setting: the
time spent by the fully evolving state $U^nW_-\varphi$ inside the time-dependent
subset $(1-L_n^*L_n)\H$ of $\H$
$$
T_2(\varphi)=\sum_{n\in\Z}
\big\langle U^nW_-\varphi,\big(1-L_n^*L_n\big)U^nW_-\varphi\big\rangle_{\H_0}.
$$
The symmetrised time delay $\tau_r^{\rm sym}(\varphi)$ for the scattering system
$(U_0,U,J)$ is then defined as the difference between the sojourn times for the fully
evolving state $U^nW_-\varphi$ and the sojourn times for the freely evolving state
$U_0^n\varphi$ before and after scattering
$$
\tau_r^{\rm sym}(\varphi)
:=\big(T_{r,1}(\varphi)+T_2(\varphi)\big)
-\tfrac12\big(T_r^0(\varphi)+T_r^0(S\varphi)\big).
$$
Our main result, properly stated in Theorem \ref{thm_sym} and Remark \ref{rem_sym}, is
the proof of the existence of the limit $\lim_{r\to\infty}\tau_r^{\rm sym}(\varphi)$
and its identity with the expectation value of a unitary analogue of the
Eisenbud-Wigner time delay operator, namely,
\begin{equation}\label{eq_sym_intro}
\lim_{r\to\infty}\tau_r^{\rm sym}(\varphi)
=\left\langle\varphi,-S^*U_0\;\!\frac{\d S}{\d U_0}\;\!\varphi\right\rangle_{\H_0}.
\end{equation}
To our knowledge, this formula is completely new. It establishes in the unitary setup
the identity between the time-dependent definition of time delay and the
time-independent definition of time delay. Under the additional assumption that the
scattering system is elastic in some appropriate sense (namely, that the operator $S$
commutes with some function of a velocity operator $V$ defined in terms of $U_0$ and
$Q$), we show in Theorem \ref{thm_non_sym} that the simpler, non-symmetrised time delay
$$
\tau_r^{\rm nsym}(\varphi)
:=\big(T_{r,1}(\varphi)+T_2(\varphi)\big)-T_r^0(\varphi)
$$
also exists in the limit $r\to\infty$ and satisfies
\begin{equation}\label{eq_non_sym_intro}
\lim_{r\to\infty}\tau_r^{\rm nsym}(\varphi)
=\lim_{r\to\infty}\tau_r^{\rm sym}(\varphi)
=\left\langle\varphi,-S^*U_0\;\!\frac{\d S}{\d U_0}\;\!\varphi\right\rangle_{\H_0}.
\end{equation}

We give now a more detailed description of the paper. In Section \ref{section_free},
we introduce the free propagator $U_0$ and the family of position operators
$Q=(Q_1,\ldots,Q_d)$, and we present the general conditions of regularity and
commutation that we impose on $U_0$ and $Q$ (Assumptions \ref{ass_regularity} and
\ref{ass_commute}). These conditions, formulated in terms of the family of unitary
operators
$$
U_0(x):=\e^{-ix\cdot Q}U_0\e^{ix\cdot Q},\quad x\in\R^d,
$$
imply the existence of a family of mutually commuting self-adjoint operators
$V=(V_1,\ldots,V_d)$ with $V_j$ given by
$$
V_j:=\hbox{s\hspace{1pt}-}\,\frac\d{\d x_j}\bigg|_{x=0}iU_0(x)U_0^{-1}
$$
on some appropriate core. The operators $V_j$ can be interpreted as the components of
the velocity vector associated to $U_0$ and $Q_j$. Accordingly, the set $\kappa(U_0)$
of values in the spectrum $\sigma(U_0)$ where $V=0$ (precisely defined in Definition
\ref{def_kappa}) plays an important role and is called the set of critical values of
$U_0$. It is a unitary analogue of the set of critial values introduced in
\cite[Def.~2.5]{RT12_0} in the self-adjoint setup.

In Section \ref{section_smooth}, we use commutator methods for unitary operators
\cite{FRT13} to construct a conjugate operator for $U_0$ and to prove a Mourre
estimate for $U_0$ on the set $\sigma(U_0)\setminus\kappa(U_0)$ (Lemmas \ref{lemma_A}
and \ref{lemma_Mourre}). As a consequence, we obtain in Theorem \ref{thm_spectrum} a
class locally $U_0$-smooth operators on $\S^1\setminus\kappa(U_0)$ and we show that
the operator $U_0$ has purely absolutely continuous in
$\sigma(U_0)\setminus\kappa(U_0)$. To illustrate these results, we present in Examples
\ref{ex_V_constant} and \ref{ex_Laplacian} the cases where the velocity vector $V$ is
constant and $U_0$ is the time-one propagator for the Laplacian in $\R^d$.

In Section \ref{section_summation}, we prove a formula which relates the evolution of
the localisation operator $f(Q)$ under $U_0$ to a time operator $T_f$. First, we
recall in Section \ref{section_averaged} the properties of averaged localisation
functions $R_f:\R^d\setminus\{0\}\to\C$ which appears naturally when dealing with
quantum time delay. Then, in Section \ref{section_proof}, we prove for appropriate
vectors $\varphi\in\H_{\rm ac}(U_0)$ the summation formula
\begin{equation}\label{eq_formula_intro}
\lim_{r\to\infty}\tfrac12\sum_{n\ge0}\big\langle\varphi,\big(U_0^{-n}f(Q/r)U_0^n
-U_0^nf(Q/r)U_0^{-n}\big)\varphi\big\rangle_{\H_0}
=\big\langle\varphi,T_f\;\!\varphi\big\rangle_{\H_0}.
\end{equation}
The proof, given in Theorem \ref{thm_summation}, is not trivial: It involves
commutator methods for families of self-adjoint and unitary operators, the class of
locally $U_0$-smooth operators obtained in Section \ref{section_smooth}, operator
identities following from the Trotter-Kato formula (Lemma \ref{lemma_trotter}) and a
repeated use of Lebesgue's dominated convergence theorem. The operator $T_f$ is
similar to an operator appearing in the self-adjoint setup \cite[Prop.~5.2]{RT12_0}.
Its precise definition is given in terms of a quadratic form (see Proposition
\ref{prop_T_f}), but formally
$$
T_f=\tfrac12\big(Q\cdot(\nabla R_f)(V)+(\nabla R_f)(V)\cdot Q\big).
$$
In Section \ref{section_interpretation}, Lemma \ref{lemma_canonical}, we show that the
operators $T_f$ and $U_0$ satisfy under appropriate conditions the relation
$$
U_0^{-1}\big[T_f,U_0\big]=-1,
$$
which is the unitary analogue of the canonical time-energy commutation relation of the
self-adjoint setup. Therefore, the operator $T_f$ can be interpreted as a time
operator for $U_0$, and $T_f$ is equal in some suitable sense to the operator
$-U_0\;\!\frac\d{\d U_0}$. Indeed, by applying Mackey's imprimitivity theorem
\cite{Ors79}, we are able to show that $T_f$ acts as the differential operator
$-z\;\!\frac\d{\d z}$ ($z\in\S^1$) in a Hilbert space isomorphic to $\H_0$ (see Remark
\ref{rem_interpretation}). In consequence, the formula \eqref{eq_formula_intro} can be
seen as an equality between on the l.h.s. the difference of times spent by the
evolving state $U_0^n\varphi$ in the future (first term) and in the past (second term)
within dilated regions defined by the localisation operators $f(Q/r)$ and on the
r.h.s. the expectation value in $\varphi$ of the time operator $T_f$. At the end of
Section \ref{section_interpretation}, we illustrate these results once again with the
cases where the velocity vector $V$ is constant and $U_0$ is the time-one propagator
for the Laplacian in $\R^d$ (see Examples \ref{V_constant_continued} and
\ref{Laplacian_continued}).

In Section \ref{section_symmetrised}, we prove \eqref{eq_sym_intro}, that is, the
existence of the symmetrised time delay $\tau_r^{\rm sym}(\varphi)$ in the limit
$r\to\infty$ and its identity with the expectation value of the unitary analogue of
the Eisenbud-Wigner time delay operator (see Theorem \ref{thm_sym} and Remark
\ref{rem_sym}). The main ingredient of the proof is the summation formula
\eqref{eq_formula_intro}. In Section \ref{section_non}, under the additional
assumption that the scattering operator $S$ commutes with some appropriate function of
the velocity operator $V$, we prove that the simpler, non-symmetrised time delay
$\tau_r^{\rm nsym}(\varphi)$ also exists in the limit $r\to\infty$ and satisfies the
same identity (that is, \eqref{eq_non_sym_intro}, see Theorem \ref{thm_non_sym}).

Before concluding, we would like to emphasize that our results here in the unitary
setup  are not a mere consequence of the corresponding results \cite{RT12_0,RT12} in
the self-adjoint setup. The standard tools allowing one to go from unitary operators
to self-adjoint operators (such as the Cayley transform, functional calculus or
operator logarithms) are not suited for the problem of quantum time delay. Even more,
various proofs in the present paper turn out to be more subtle than the corresponding
proofs in the self-adjoint setup. Some of the reasons explaining this fact are the
following:
\begin{enumerate}
\item[$\bullet$] In the unitary setup, the sojourn times are defined as sums over a
discrete time $n\in\Z$, whereas in the self-adjoint setup the sojourn times are
defined as integrals over a continuous time $t\in\R$. So, in order to obtain results
in the unitary setup, one has to evaluate infinite sums, which in general is more
challenging than to evaluate improper integrals. In particular, the proof of the
summation formula \eqref{eq_formula_intro} is more technical than the proof of the
corresponding integral formula of the self-adjoint setup \cite[Thm.~5.5]{RT12_0}, and
the proof of the existence of the non-symmetrised time delay \eqref{eq_non_sym_intro}
relies on a preliminary result based on the Poisson summmation formula (Lemma
\ref{lemma_poisson}) not needed in the self-adjoint setup.
\item[$\bullet$] The unitary operators $U_0$ and $U$ that we consider are completely
general. In particular, there are not supposed to be time-one propagators of some
self-adjoint operators $H_0$ and $H$. In consequence, one cannot apply all the
technics coming from the self-adjoint theory. Moreover, one does not have at disposal
a predefined dense set $\D\subset\H_0$ (such as the domain of $H_0$) where to perform
the necessary the calculations for the free theory. Instead, one has to come up with
an assumption on $U_0$ specific enough to put into evidence a dense set of
$\D\subset\H_0$ appropriate for the calculations, but general enough not to
oversimplify the theory. Assumption \ref{ass_regularity} fullfils these requirements.
\item[$\bullet$] In the self-adjoint setup, the unitary groups generated by the free
Hamiltonian and the time operator satisfy in favorable situations the Weyl relation.
Thus, one can apply Stone-von Neumann theorem to conclude that the time operator acts
as the energy derivative in the spectral representation of the free Hamiltonian (see
\cite[Sec.~6]{RT12_0}). In the unitary setup, the unitary groups generated by the free
propagator and the time operator satisfy at best only an imprimitivity relation. Thus,
one has to apply Mackey's imprimitivity theorem, which is more complex than Stone-von
Neumann theorem, to conclude that the time operator acts as the energy derivative in
the spectral representation of the free propagator (see Section
\ref{section_interpretation}).
\end{enumerate}

To conclude, we point out that the theory presented here is general, adapted to cover
a variety of unitary scattering systems, both in the one and two-Hilbert spaces
setting. Therefore, we plan in the future to apply it to various unitary scattering
systems, as for instance anisotropic quantum walks as presented in
\cite{RST_1,RST_2}.\\

\noindent
{\bf Acknowledgements.} The authors express their gratitude to the referees who
dedicated time to review this manuscript and suggested various improvements to
the text. In particular, we are really grateful to one referee who pointed to us possible
alternative proofs for some of our results.

\section{Free propagator and position operators}\label{section_free}
\setcounter{equation}{0}

In this section, we recall needed facts on commutators methods, we introduce our
assumptions on the free propagator and the position operators, and we describe a set
of critical values of the free propagator which appears naturally under our
assumptions.

We start by recalling some facts on commutators methods borrowed from
\cite{ABG96,GG99}. Let $H$ be a self-adjoint operator with domain $\dom(H)$ and
spectrum $\sigma(H)$ in a Hilbert space $\H_0$, and let $A$ be a second self-adjoint
operator with domain $\dom(A)$ in $\H_0$. We say that $H$ is of class $C^k(A)$ with
$k\in\N$ if for some $\omega\in\C\setminus\sigma(H)$ the map
\begin{equation}\label{condition_Ck_A}
\R\ni t\mapsto\e^{-itA}(H-\omega)^{-1}\e^{itA}\in\B(\H_0)
\end{equation}
is strongly of class $C^k$. In the case $k=1$, the quadratic form
$$
\dom(A)\ni\varphi\mapsto
\big\langle\varphi,(H-\omega)^{-1}A\;\!\varphi\big\rangle_{\H_0}
-\big\langle A\;\!\varphi,(H-\omega)^{-1}\varphi\big\rangle_{\H_0}\in\C
$$
extends continuously to a bounded operator denoted by $\big[(H-\omega)^{-1},A\big]$.
Furthermore, the set $\dom(H)\cap\dom(A)$ is a core for $H$ and the quadratic form
$$
\dom(H)\cap\dom(A)\ni\varphi\mapsto
\big\langle H\varphi,A\;\!\varphi\big\rangle_{\H_0}
-\big\langle A\;\!\varphi,H\varphi\big\rangle_{\H_0}\in\C
$$
is continuous in the topology of $\dom(H)$. Thus, it extends uniquely to a continuous
quadratic form $[H,A]$ on $\dom(H)$ which can be identified with a continuous operator
from $\dom(H)$ to the adjoint space $\dom(H)^*$, and the following equality holds:
\begin{equation}\label{eq_resolvent}
\big[(H-\omega)^{-1},A\big]=-(H-\omega)^{-1}[H,A](H-\omega)^{-1}.
\end{equation}
In \cite[Lemma~2]{GG99}, it has been shown that if $[H,A]\dom(H)\subset\H$, then
$\e^{itA}\dom(H)\subset\dom(H)$ for each $t\in\R$. Accordingly, we say in the sequel
that $i[H,A]$ is essentially self-adjoint on $\dom(H)$ if $[H,A]\;\!\dom(H)\subset\H$
and if $i[H,A]$ is essentially self-adjoint on $\dom(H)$ in the usual sense.

Now, let $Q:=(Q_1,\ldots,Q_d)$ be a family of mutually strongly commuting self-adjoint
operators in $\H_0$ (the position operators). Then, any function $f\in\linf(\R^d)$
defines by $d$-variables functional calculus a bounded operator $f(Q)$ in $\H_0$. In
particular, the operator $\e^{ix\cdot Q}$, with $x\cdot Q:=\sum_{j=1}^dx_jQ_j$, is
unitary in $\H_0$ for each $x\in\R^d$. In this context, we say that $H$ is of class
$C^k(Q)$ with $k\in\N$ if for some $\omega\in\C\setminus\sigma(H)$ the map
\begin{equation}\label{condition_Ck_Q}
\R^d\ni x\mapsto\e^{-ix\cdot Q}(H-\omega)^{-1}\e^{ix\cdot Q}\in\B(\H_0)
\end{equation}
is strongly of class $C^k$. Clearly, if $H$ is of class $C^k(Q)$, then $H$ is of class
$C^k(Q_j)$ for each $j$.

\begin{Remark}
The definitions are similar if we consider a bounded operator $B\in\B(\H_0)$ instead
of a self-adjoint operator $H$. In such a case, we use the notation $B\in C^k(A)$ if
the map \eqref{condition_Ck_A}, with $(H-\omega)^{-1}$ replaced by $B$, is strongly of
class $C^k$, and we use the notation $B\in C^k(Q)$ if the map \eqref{condition_Ck_Q},
with $(H-\omega)^{-1}$ replaced by $B$, is strongly of class $C^k$.
\end{Remark}

In the sequel, we assume the existence of a unitary operator $U_0\in\B(\H_0)$ (the
free propagator) with associated family of unitary operators
$$
U_0(x):=\e^{-ix\cdot Q}U_0\e^{ix\cdot Q},\quad x\in\R^d,
$$
regular with respect to $Q$ in the following sense:

\begin{Assumption}[Regularity]\label{ass_regularity}
The map
$$
\R^d\ni x\mapsto U_0(x)U_0^{-1}\in\B(\H_0)
$$
is strongly differentiable on a core $\D\subset\H_0$ of the operator $Q^2$, and for
each $j\in\{1,\ldots,d\}$ the operator
$$
V_j\;\!\varphi:=\emph{s\hspace{1pt}-}\,\frac\d{\d x_j}\bigg|_{x=0}
iU_0(x)U_0^{-1}\varphi,\quad\varphi\in\D,
$$
is essentially self-adjoint, with self-adjoint extension denoted by the same symbol.
The operator $V_j$ is of class $C^1(Q)$, and for each $k\in\{1,\ldots,d\}$,
$i\big[V_j,Q_k\big]$ is essentially self-adjoint on $\dom(V_j)$, with self-adjoint
extension denoted by $V'_{jk}$. The operator $V'_{jk}$ is of class $C^1(Q)$, and for
each $\ell\in\{1,\ldots,d\}$, $i\big[V'_{jk},Q_\ell\big]$ is essentially self-adjoint
on $\dom\big(V'_{jk}\big)$, with self-adjoint extension denoted by $V''_{jk\ell}$.
\end{Assumption}

Assumption \ref{ass_regularity} implies that the set $\D$ is a core for all the
operators $V_j$, $V'_{jk}$, and $V''_{jk\ell}$. Assumption \ref{ass_regularity} also
implies the invariance of the domains $\dom(V_j)$ under the action of the unitary
group $\{\e^{ix\cdot Q}\}_{x\in\R^d}$. Indeed, the condition
$\big[V_j,Q_k\big]\dom(V_j)\subset\H$ and \cite[Lemma~2]{GG99} imply that
$\e^{itQ_k}\dom(V_j)\subset\dom(V_j)$ for each $t\in\R$. Thus
$\e^{itQ_k}\dom(V_j)=\dom(V_j)$ for each $t\in\R$, and since this holds for each
$k\in\{1,\ldots,d\}$ one obtains that $\e^{ix\cdot Q}\dom(V_j)=\dom(V_j)$ for each
$x\in\R^d$. As a consequence, the operators
$$
V_j(x):=\e^{-ix\cdot Q}V_j\e^{ix\cdot Q},\quad x\in\R^d,
$$
are self-adjoint operators with domains $\dom\big(V_j(x)\big)=\dom(V_j)$. Similarly,
the domains $\dom\big(V'_{jk}\big)$ are left invariant by the unitary group
$\{\e^{ix\cdot Q}\}_{x\in\R^d}$, and the operators
$$
V'_{jk}(x):=\e^{-ix\cdot Q}V'_{jk}\e^{ix\cdot Q},\quad x\in\R^d,
$$
are self-adjoint operators with domains
$\dom\big(V'_{jk}(x)\big)=\dom\big(V'_{jk}\big)$.

\begin{Remark}
The operators $V_j$ and $V'_{jk}$ can be interpreted as the components of the velocity
vector and the acceleration matrix associated to the propagator $U_0$ and the position
operators $Q_j$.
\end{Remark}

Our second main assumption on the operators $U_0(x)$ is a commutation assumption:

\begin{Assumption}[Commutation]\label{ass_commute}
$\big[U_0(x),U_0(y)\big]=0$ for all $x,y\in\R^d$.
\end{Assumption}

Assumption \ref{ass_commute} implies that the operators $U_0(x)$ mutually commute in
the strong sense, namely, if $E^{U_0(x)}$ denotes the spectral measure of $U_0(x)$ on
the complex unit circle $\S^1:=\{\e^{i\theta}\mid\theta\in[0,2\pi)\}$, then
$$
\big[E^{U_0(x)}(\Theta),E^{U_0(y)}(\Theta')\big]=0
$$
for all $x,y\in\R^d$ and all Borel sets $\Theta,\Theta'\subset\S^1$ (see
\cite[Prop.~5.27]{Sch12}).

Additional commutation relations are obtained in the following lemma:

\begin{Lemma}\label{lemma_commute}
Let Assumptions \ref{ass_regularity} and \ref{ass_commute} be satisfied. Then, the
operators $U_0(x)$, $\big(V_j(y)+i\big)^{-1}$, $\big(V'_{k\ell}(z)+i\big)^{-1}$
mutually commute for all $x,y,z\in\R^d$ and all $j,k,\ell\in\{1,\ldots,d\}$.
\end{Lemma}

\begin{proof}
Let $j,k,\ell,m\in\{1,\ldots,d\}$, $x,y\in\R^d$, and set
$$
R^{V_j(x)}:=\big(V_j(x)+i\big)^{-1}
\quad\hbox{and}\quad
R^{V_{k\ell}'(x)}:=\big(V_{k\ell}'(x)+i\big)^{-1}.
$$
For any $\varphi,\psi\in\H_0$, there exist
$(\varphi_n)_{n\in\N},(\psi_n)_{n\in\N}\subset\D$ such that
$$
\lim_{n\to\infty}\big\|\varphi_n-\big(R^{V_j}\big)^*\varphi\big\|_{\dom(V_j)}=0
\quad\hbox{and}\quad
\lim_{n\to\infty}\big\|\psi_n-R^{V_j}\psi\big\|_{\dom(V_j)}=0.
$$
This, together with the facts that $\big[U_0(x),iU_0(y)U_0^{-1}\big]=0$ and
$
\hbox{s\hspace{1pt}-}\,\frac\d{\d y_j}\big|_{y=0}\big(iU_0(y)U_0^{-1}\big)^*\varphi_m
=V_j\varphi_m
$,
implies that
\begin{align}
\left\langle\varphi,R^{V_j}U_0(x)\;\!V_jR^{V_j}\psi\right\rangle_{\H_0}
&=\lim_{m\to\infty}\lim_{n\to\infty}
\big\langle\varphi_m,U_0(x)\;\!V_j\;\!\psi_n\big\rangle_{\H_0}\nonumber\\
&=\lim_{m\to\infty}\lim_{n\to\infty}\frac\d{\d y_j}\bigg|_{y=0}
\left\langle\big(iU_0(y)U_0^{-1}\big)^*\varphi_m,U_0(x)\psi_n
\right\rangle_{\H_0}\nonumber\\
&=\lim_{m\to\infty}\lim_{n\to\infty}
\big\langle V_j\;\!\varphi_m,U_0(x)\psi_n\big\rangle_{\H_0}\nonumber\\
&=\left\langle\varphi,R^{V_j}V_j\;\!U_0(x)R^{V_j}
\psi\right\rangle_{\H_0}.\label{eq_calculation}
\end{align}
Since $\varphi,\psi$ are arbitrary, this implies that
$$
\big[U_0(x),R^{V_j}\big]=-R^{V_j}\big[U_0(x),V_j\big]R^{V_j}=0.
$$
Therefore, we obtain that
$$
\big[U_0(x),R^{V_j(y)}\big]
=\e^{-iy\cdot Q}\big[U_0(x-y),R^{V_j}\big]\e^{iy\cdot Q}
=0
$$
and thus the operators $U_0(x)$ and $R^{V_j(y)}$ commute.

A calculation as in \eqref{eq_calculation} using the commutation of $U_0(x)U_0^{-1}$
and $R^{V_j(y)}$ implies that
$$
\left\langle\varphi,R^{V_j}V_jR^{V_k(y)}R^{V_j}
\psi\right\rangle_{\H_0}
=\left\langle\varphi,R^{V_j}R^{V_k(y)}V_jR^{V_j}
\psi\right\rangle_{\H_0}.
$$
Since $\varphi,\psi$ are arbitrary, this implies that
$$
\big[R^{V_j},R^{V_k(y)}\big]
=-R^{V_j}\big[V_j,R^{V_k(y)}\big]R^{V_j}
=0.
$$
Therefore, we obtain that
$$
\big[R^{V_j(x)},R^{V_k(y)}\big]
=\e^{-ix\cdot Q}\big[R^{V_j},R^{V_k(y)}\big]\e^{ix\cdot Q}
=0,
$$
and thus the operators $R^{V_j(x)}$ and $R^{V_k(y)}$ commute.	

Let $e_\ell$ be the $\ell$-th standard orthonormal vector in $\R^d$. Then, the
commutation of $R^{V_j(x)}$ and $R^{V_k(y)}$ implies that
$$
R^{V_j(x)}\frac{R^{V_k(y+\varepsilon e_\ell)}
-R^{V_k(y)}}\varepsilon\\
=\frac{R^{V_k(y+\varepsilon e_\ell)}
-R^{V_k(y)}}\varepsilon R^{V_j(x)},
\quad\varepsilon\in\R\setminus\{0\}.
$$
Taking the limit $\varepsilon\to0$ and using \eqref{eq_resolvent} and the strong
commutation of $V_j(x)$ and $V_k(y)$, one obtains
$$
R^{V_k(y)}\big[R^{V_j(x)},V_{k\ell}'(y)\big]R^{V_k(y)}=0.
$$
Since the resolvent $R^{V_k(y)}$ on the left is injective, this implies that
$\big[R^{V_j(x)},V_{k\ell}'(y)\big]=0$ on $\dom\big(V_k(x)\big)$, and since
$\dom\big(V_k(x)\big)$ is a core for $V_{k\ell}'(y)$ the last equality extends to
$\dom\big(V_{k\ell}'(y)\big)$. Therefore, we obtain that
$$
\big[R^{V_j(x)},R^{V_{k\ell}'(y)}\big]
=-R^{V_{k\ell}'(y)}\big[R^{V_j(x)},V_{k\ell}'(y)\big]
R^{V_{k\ell}'(y)}
=0,
$$
and thus the operators $R^{V_j(x)}$ and $R^{V_{k\ell}'(y)}$ commute.

The commutation of $U_0(x)$ and $R^{V_k(y)}$ implies that
$$
U_0(x)\;\!\frac{R^{V_k(y+\varepsilon e_\ell)}
-R^{V_k(y)}}\varepsilon\\
=\frac{R^{V_k(y+\varepsilon e_\ell)}
-R^{V_k(y)}}\varepsilon\;\!U_0(x),
\quad\varepsilon\in\R\setminus\{0\}.
$$
Taking the limit $\varepsilon\to0$, and using \eqref{eq_resolvent}, the commutation of
$U_0(x)$ and $R^{V_k(y)}$, and the commutation of $R^{V_k(y)}$ y $R^{V_{k\ell}'(y)}$,
one obtains that
\begin{align*}
&R^{V_k(y)}\big[U_0(x),V_{k\ell}'(y)\big]R^{V_k(y)}=0\\
&\iff R^{V_k(y)}R^{V_{k\ell}'(y)}\big[U_0(x),V_{k\ell}'(y)\big]
R^{V_{k\ell}'(y)} R^{V_k(y)}=0\\
&\iff R^{V_{k\ell}'(y)}\big[U_0(x),V_{k\ell}'(y)\big]
R^{V_{k\ell}'(y)}=0.
\end{align*}
Therefore, we obtain that
$$
\big[U_0(x),R^{V_{k\ell}'(y)}\big]
=-R^{V_{k\ell}'(y)}\big[U_0(x),V_{k\ell}'(y)\big]
R^{V_{k\ell}'(y)}=0,
$$
and thus the operators $U_0(x)$ and $R^{V_{k\ell}'(y)}$ commute.

Finally, the commutation of $R^{V_{jk}'(x)}$ and $R^{V_{\ell m}'(y)}$ is proved in a
similar way. The details are left to the reader.
\end{proof}

Lemma \ref{lemma_commute} and \cite[Prop.~5.27]{Sch12} imply that the operators
$U_0(x)$, $V_j(y)$ and $V_{k\ell}'(z)$ mutually commute in the strong sense for all
$x,y,z\in\R^d$ and all $j,k,\ell\in\{1,\ldots,d\}$, that is, the spectral projections
of $U_0(x)$, $V_j(y)$ and $V_{k\ell}'(z)$ mutually commute for all $x,y,z\in\R^d$ and
all $j,k,\ell\in\{1,\ldots,d\}$. Using these commutation relations, we can establish
other useful facts:

\begin{Lemma}\label{lemma_V}
Let Assumptions \ref{ass_regularity} and \ref{ass_commute} be satisfied.
\begin{enumerate}
\item[(a)] For all $j\in\{1,\ldots,d\}$, $x\in\R^d$ and $\varphi\in\D$, one has 
$$
\emph{s\hspace{1pt}-}\,\frac\d{\d x_j}iU_0(x)U_0^{-1}\varphi
=V_j(x)U_0(x)U_0^{-1}\varphi.
$$
\item[(b)] For all $n\in\Z$ and $j\in\{1,\ldots,d\}$, one has 
$$
U_0^n\big(\dom(Q_j)\cap\dom(V_j)\big)=\dom(Q_j)\cap\dom(V_j).
$$
\item[(c)] For all $n\in\Z$, $j\in\{1,\ldots,d\}$ and
$\varphi\in\dom(Q_j)\cap\dom(V_j)$, one has 
$$
U_0^nQ_jU_0^{-n}\varphi=(Q_j-n\;\!V_j)\varphi,
$$
and the operator $(Q_j-n\;\!V_j)$ is essentially self-adjoint on
$\dom(Q_j)\cap\dom(V_j)$.
\item[(d)] For all $x\in\R^d$ and $n\in\Z$, the operator
$\big(x\cdot Q-n\;\!(x\cdot V)\big)$ is essentially self-adjoint on
$\dom(x\cdot Q)\cap\dom(x\cdot V)$, with self-adjoint extension
$$
\overline{\big(x\cdot Q-n\;\!(x\cdot V)\big)
\upharpoonright\dom(x\cdot Q)\cap\dom(x\cdot V)}
=U_0^n(x\cdot Q)U_0^{-n}.
$$
\item[(e)] For all $j,k\in\{1,\ldots,d\}$, one has $V_{jk}'=V_{kj}'$.
\end{enumerate}
\end{Lemma}

The result of point (c) has the following interpretation: After time $n$, the position
$U_0^nQ_jU_0^{-n}$ of the quantum system with propagator $U_0$ is equal to the value
$Q$ of its initial position minus $n$ times the value $V$ of its initial velocity.

\begin{proof}
(a) Let $\varphi\in\D$ and $\psi\in\e^{-ix\cdot Q}\D$. Then, one has
\begin{align*}
\big\langle\psi,
\hbox{s\hspace{1pt}-}\tfrac\d{\d x_j}iU_0(x)U_0^{-1}\varphi\big\rangle_{\H_0}
&=\lim_{\varepsilon\to0}\left\langle\psi,
\tfrac i\varepsilon\big(\e^{-ix\cdot Q}U_0(\varepsilon e_j)\e^{ix\cdot Q}U_0^{-1}
-\e^{-ix\cdot Q}U_0\e^{ix\cdot Q}U_0^{-1}\big)\varphi\right\rangle_{\H_0}\\
&=\lim_{\varepsilon\to0}\left\langle-\tfrac i\varepsilon
\big(U_0(\varepsilon e_j)U_0^{-1}-1\big)^*\e^{ix\cdot Q}\psi,
U_0\e^{ix\cdot Q}U_0^{-1}\varphi\right\rangle_{\H_0}.
\end{align*}
Since $\e^{ix\cdot Q}\psi\in\D$, Assumption \ref{ass_regularity} implies that
$$
-\slim_{\varepsilon\to0}\tfrac i\varepsilon\big(U_0(\varepsilon e_j)U_0^{-1}-1\big)^*
\e^{ix\cdot Q}\psi=V_j\e^{ix\cdot Q}\psi.
$$
Thus,
$$
\big\langle\psi,
\hbox{s\hspace{1pt}-}\tfrac\d{\d x_j}iU_0(x)U_0^{-1}\varphi\big\rangle_{\H_0}
=\big\langle V_j\e^{ix\cdot Q}\psi,U_0\e^{ix\cdot Q}U_0^{-1}\varphi\big\rangle_{\H_0}
=\big\langle\psi,V_j(x)U_0(x)U_0^{-1}\varphi\big\rangle_{\H_0}
$$
with $U_0(x)U_0^{-1}\varphi\in\dom\big(V_j(x)\big)=\dom(V_j)$ due to the commutation
of $V_j(x)$ and $U_0(x)U_0^{-1}$. Since $\D$ is dense in $\H_0$ and
$\e^{-ix\cdot Q}:\H_0\to\H_0$ is a homeomorphism, the set of vectors $\psi$ is dense
in $\H_0$, and thus
$$
\emph{s\hspace{1pt}-}\tfrac\d{\d x_j}iU_0(x)U_0^{-1}\varphi
=V_j(x)U_0(x)U_0^{-1}\varphi.
$$

(b) Let $\varphi\in\dom(Q_j)\cap\dom(V_j)$ and $t\in\R\setminus\{0\}$. Then,
$$
\big(\e^{-itQ_j}-1\big)U_0^{-1}\varphi
=U(te_j)^{-1}\big(\e^{-itQ_j}-1\big)\varphi+\big(U(te_j)^{-1}-U_0^{-1}\big)\varphi,
$$
and point (a) and Lemma \ref{lemma_V} imply for each $\psi\in\D$ that
\begin{align*}
\big\langle\psi,\big(U(te_j)^{-1}-U_0^{-1}\big)\varphi\big\rangle_{\H_0}
&=\big\langle\big(U(te_j)U_0^{-1}-1\big)\psi,U_0^{-1}\varphi\big\rangle_{\H_0}\\
&=\left\langle\int_0^t\d s\;\!\tfrac\d{\d s}\;\!U(se_j)U_0^{-1}\psi,
U_0^{-1}\varphi\right\rangle_{\H_0}\\
&=\left\langle-i\int_0^t\d s\;\!V_j(se_j)U(se_j)U_0^{-1}\psi,
U_0^{-1}\varphi\right\rangle_{\H_0}\\
&=\left\langle\psi,i\int_0^t\d s\;\!U(se_j)^{-1}V_j(se_j)\varphi\right\rangle_{\H_0}.
\end{align*}
Since $\D$ is dense in $\H_0$, it follows that
$$
\big(\e^{-itQ_j}-1\big)U_0^{-1}\varphi
=U(te_j)^{-1}\big(\e^{-itQ_j}-1\big)\varphi
+i\int_0^t\d s\;\!U(se_j)^{-1}V_j(se_j)\varphi,
$$
and thus that
$$
Q_jU_0^{-1}\varphi
=\slim_{t\to0}\tfrac it\big(\e^{-itQ_j}-1\big)U_0^{-1}\varphi
=U_0^{-1}Q_j\varphi-U_0^{-1}V_j\varphi.
$$
Therefore, we have shown that
$U_0^{-1}\big(\dom(Q_j)\cap\dom(V_j)\big)\subset\dom(Q_j)$ with
$$
U_0Q_jU_0^{-1}\varphi=\big(Q_j-V_j\big)\varphi
$$
for all $\varphi\in\dom(Q_j)\cap\dom(V_j)$. Since $U_0^{-1}$ and $V_j$ commute, we
also have $U_0^{-1}\big(\dom(Q_j)\cap\dom(V_j)\big)\subset\dom(V_j)$, and thus
\begin{equation}\label{inclusion_one}
U_0^{-1}\big(\dom(Q_j)\cap\dom(V_j)\big)\subset\dom(Q_j)\cap\dom(V_j).
\end{equation}

Starting with the expression $(\e^{-itQ_j}-1)U_0\varphi$, we can show with similar
arguments that
\begin{equation}\label{inclusion_two}
U_0\big(\dom(Q_j)\cap\dom(V_j)\big)\subset\dom(Q_j)\cap\dom(V_j)
\end{equation}
and that
$$
U_0^{-1}Q_jU_0\varphi=\big(Q_j+V_j\big)\varphi
$$
for all $\varphi\in\dom(Q_j)\cap\dom(V_j)$. Using \eqref{inclusion_one} and
\eqref{inclusion_two} we get
$$
\dom(Q_j)\cap\dom(V_j)
=U_0U_0^{-1}\big(\dom(Q_j)\cap\dom(V_j)\big)
\subset U_0\big(\dom(Q_j)\cap\dom(V_j)\big)
\subset\big(\dom(Q_j)\cap\dom(V_j)\big).
$$
Thus, we obtain
$U_0\big(\dom(Q_j)\cap\dom(V_j)\big)=\big(\dom(Q_j)\cap\dom(V_j)\big)$, which in turn
implies for each $n\in\Z$ that
$$
U_0^n\big(\dom(Q_j)\cap\dom(V_j)\big)=\dom(Q_j)\cap\dom(V_j).
$$

(c) We prove the first claim by induction on $n\ge0$ (the case $n\le0$ is similar).
The case $n=0$ is trivial, the case $n=1$ has been shown in the proof of point (b),
and in the case $n-1\ge1$ we assume that the claim is true. Then, to prove the claim
in the case $n$, we take $\varphi\in\dom(Q_j)\cap\dom(V_j)$ and use successively the
fact that $U_0^{-1}\big(\dom(Q_j)\cap\dom(V_j)\big)=\dom(Q_j)\cap\dom(V_j)$, the
induction hypothesis, the commutation of $V_j$ and $U_0$, and the claim in the case
$n=1$ to obtain the equalities
\begin{align*}
U_0^nQ_jU_0^{-n}\varphi
&=U_0\big(Q_j-(n-1)V_j\big)U_0^{-1}\varphi\\
&=U_0Q_jU_0^{-1}\varphi-(n-1)V_j\varphi\\
&=\big(Q_j-V_j\big)\varphi-(n-1)V_j\varphi\\
&=(Q_j-n\;\!V_j)\varphi.
\end{align*}

The second claim follows from the first claim if one takes into account the fact that
$\D$ is a core for $Q_j$ and the inclusions
$$
\dom(Q_j)\cap\dom(V_j)=U_0^n\big(\dom(Q_j)\cap\dom(V_j)\big)\supset U_0^n\D,
$$
which follow from point (b) and the definition of the set $\D$.

(d) Point (c) implies that
\begin{equation}\label{eq_xQ}
U_0^n(x\cdot Q)U_0^{-n}=\big(x\cdot Q-n\;\!(x\cdot V)\big)
\quad\hbox{on}\quad\cap_{j=1}^d\big(\dom(Q_j)\cap\dom(V_j)\big).
\end{equation}
Furthermore, point (b) implies that
$$
\cap_{j=1}^d\big(\dom(Q_j)\cap\dom(V_j)\big)
=\cap_{j=1}^dU_0^n\big(\dom(Q_j)\cap\dom(V_j)\big)
\supset U_0^n\D.
$$
Since $\D$ is a core for $x\cdot Q$ (because $\D$ is dense in
$\dom(Q^2)\subset\dom(x\cdot Q)$), it follows from \eqref{eq_xQ} that the operator
$\big(x\cdot Q-n\;\!(x\cdot V)\big)$ is essentially self-adjoint on
$\cap_{j=1}^d\big(\dom(Q_j)\cap\dom(V_j)\big)$, and thus essentially self-adjoint on
$\dom(x\cdot Q)\cap\dom(x\cdot V)$. This, together with the uniqueness of the
self-adjoint extension of an essentially self-adjoint operator, implies that
\begin{align*}
\overline{\big(x\cdot Q-n\;\!(x\cdot V)\big)
\upharpoonright\dom(x\cdot Q)\cap\dom(x\cdot V)}
&=\overline{\big(x\cdot Q-n\;\!(x\cdot V)\big)
\upharpoonright\cap_{j=1}^d\big(\dom(Q_j)\cap\dom(V_j)\big)}\\
&=\overline{U_0^n(x\cdot Q)U_0^{-n}
\upharpoonright\cap_{j=1}^d\big(\dom(Q_j)\cap\dom(V_j)\big)}\\
&=U_0^n(x\cdot Q)U_0^{-n}.
\end{align*}

(e) Let $\varphi,\psi\in\D$. Using the facts that
$\D\subset\dom(V_j)\subset\dom\big(V_{jk}'\big)$, that
$V_{jk}'\upharpoonright\D=i[V_j,Q_k]\upharpoonright\D$, that
$\D\subset\dom(Q_j)\cap\dom(V_j)$ and point (b), we obtain the equalities
\begin{align}
\big\langle\varphi,V_{jk}'\psi\big\rangle_{\H_0}
&=\big\langle V_j\varphi,iQ_k\psi\big\rangle_{\H_0}
-\big\langle Q_k\varphi,iV_j\psi\big\rangle_{\H_0}\nonumber\\
&=\tfrac\d{\d t}\big|_{t=0}
\left(\big\langle U_0(te_j)U_0^{-1}\varphi,Q_k\psi\big\rangle_{\H_0}
+\big\langle Q_k\varphi,U_0(te_j)U_0^{-1}\psi\big\rangle_{\H_0}\right)\nonumber\\
&=\tfrac\d{\d t}\big|_{t=0}\left(\big\langle\e^{-itQ_j}U_0\e^{itQ_j}U_0^{-1}
\varphi,Q_k\psi\big\rangle_{\H_0}+\big\langle Q_k\varphi,\e^{-itQ_j}U_0
\e^{itQ_j}U_0^{-1}\psi\big\rangle_{\H_0}\right)\nonumber\\
&=\big\langle iU_0Q_jU_0^{-1}\varphi,Q_k\psi\big\rangle_{\H_0}
+\big\langle-iQ_j\varphi,Q_k\psi\big\rangle_{\H_0}
+\big\langle Q_k\varphi,iU_0Q_jU_0^{-1}\psi\big\rangle_{\H_0}\nonumber\\
&\quad+\big\langle Q_k\varphi,-iQ_j\psi\big\rangle_{\H_0},\label{eq_V_jk_1}
\end{align}
with
$$
\big\langle-iQ_j\varphi,Q_k\psi\big\rangle_{\H_0}
=\tfrac\d{\d t}\big|_{t=0}\big\langle\e^{-itQ_j}\varphi,Q_k\psi\big\rangle_{\H_0}
=\big\langle Q_k\varphi,iQ_j\psi\big\rangle_{\H_0}.
$$
Thus, using the fact that
$
\hbox{s\hspace{1pt}-}\tfrac\d{\d t}\big|_{t=0}\big(U_0(-te_k)U_0^{-1}\big)^*
=-i\;\!V_k
$
on $\D$, we obtain
\begin{align}
\eqref{eq_V_jk_1}
&=\big\langle Q_jU_0^{-1}\varphi,-iU_0^{-1}Q_k\psi\big\rangle_{\H_0}
+\big\langle -iU_0^{-1}Q_k\varphi,Q_jU_0^{-1}\psi\big\rangle_{\H_0}\nonumber\\
&=\tfrac\d{\d t}\big|_{t=0}\left(\big\langle Q_jU_0^{-1}\varphi,
U_0^{-1}\e^{-itQ_k}\psi\big\rangle_{\H_0}+\big\langle U_0^{-1}\e^{-itQ_k}
\varphi,Q_jU_0^{-1}\psi\big\rangle_{\H_0}\right)\nonumber\\
&=\tfrac\d{\d t}\big|_{t=0}\left(\big\langle Q_jU_0^{-1}\varphi,\e^{-itQ_k}
U_0^{-1}\big(U_0(-te_k)U_0^{-1}\big)^*\psi\big\rangle_{\H_0}\right.\nonumber\\
&\hspace{50pt}\left.+\big\langle\e^{-itQ_k}U_0^{-1}\big(U_0(-te_k)U_0^{-1}\big)^*
\varphi,Q_jU_0^{-1}\psi\big\rangle_{\H_0}\right)\nonumber\\
&=\big\langle Q_jU_0^{-1}\varphi,-iQ_kU_0^{-1}\psi\big\rangle_{\H_0}
+\big\langle Q_jU_0^{-1}\varphi,-iU_0^{-1}V_k\psi\big\rangle_{\H_0}
+\big\langle-iQ_kU_0^{-1}\varphi,Q_jU_0^{-1}\psi\big\rangle_{\H_0}\nonumber\\
&\quad+\big\langle-iU_0^{-1}V_k\varphi,Q_jU_0^{-1}\psi\big\rangle_{\H_0},
\label{eq_V_jk_2}
\end{align}
with
$$
\big\langle Q_jU_0^{-1}\varphi,-iQ_kU_0^{-1}\psi\big\rangle_{\H_0}
=\tfrac\d{\d t}\big|_{t=0}\big\langle Q_jU_0^{-1}\varphi,
\e^{-itQ_k}U_0^{-1}\psi\big\rangle_{\H_0}
=\big\langle iQ_kU_0^{-1}\varphi,Q_jU_0^{-1}\psi\big\rangle_{\H_0}.
$$
Therefore, using the fact that $U_0^{-1}$, $V_k$ and $V_{kj}'$ strongly commute, we
obtain
\begin{align*}
\eqref{eq_V_jk_2}
&=\big\langle Q_jU_0^{-1}\varphi,-iV_kU_0^{-1}\psi\big\rangle_{\H_0}
+\big\langle-iV_kU_0^{-1}\varphi,Q_jU_0^{-1}\psi\big\rangle_{\H_0}\\
&=\big\langle U_0^{-1}\varphi,i[V_k,Q_j]U_0^{-1}\psi\big\rangle_{\H_0}\\
&=\big\langle U_0^{-1}\varphi,V_{kj}'U_0^{-1}\psi\big\rangle_{\H_0}\\
&=\big\langle\varphi,V_{kj}'\psi\big\rangle_{\H_0}.
\end{align*}
Since $\D$ is dense in $\H_0$, it follows that $V_{jk}'\psi=V_{kj}'\psi$. Thus,
$V_{jk}'=V_{kj}'$ on $\D$, and since $\D$ is a core for $V_{jk}'$ and $V_{kj}'$, we
conclude that $V_{jk}'=V_{kj}'$.
\end{proof}

In rest of the section, we introduce and describe a set of critical values of $U_0$
which appears naturally under our assumptions. For this, we use the notation
$V:=(V_1,\ldots,V_d)$ for the velocity vector operator, for each measurable function
$f:\R^d\to\C$ we define the operator $f(V)$ by using the $d$-variables functional
calculus, and we use the shorthand notation
$$
E^{U_0}(\lambda;\delta)
:=E^{U_0}\big(\{\e^{i\theta}\mid\theta\in(\lambda-\delta,\lambda+\delta)\}\big),
\quad\lambda\in[0,2\pi),~\delta>0.
$$

\begin{Definition}[Critical values of $U_0$]\label{def_kappa}
A number $\e^{i\lambda}\in\S^1$ is called a regular value of $U_0$ if there exists
$\delta>0$ such that
\begin{equation}\label{eq_kappa}
\lim_{\varepsilon\searrow0}\big\|\big(V^2\big(V^2+1\big)^{-1}+\varepsilon\big)^{-1}
E^{U_0}(\lambda;\delta)\big\|_{\B(\H_0)}<\infty.
\end{equation}
A number $\e^{i\lambda}\in\S^1$ that is not a regular value of $U_0$ is called a
critical value of $U_0$, and we denote by $\kappa(U_0)$ the set of critical values of
$U_0$.
\end{Definition}

\begin{Lemma}\label{lemma_kappa}
Let Assumptions \ref{ass_regularity} and \ref{ass_commute} be satisfied.
\begin{enumerate}
\item[(a)] $\kappa(U_0)$ is closed.
\item[(b)] The limit
$
\lim_{\varepsilon\searrow0}
\big\|\big(V^2\big(V^2+1\big)^{-1}+\varepsilon\big)^{-1}E^{U_0}(\Theta)\big\|_{\B(\H_0)}
$
is finite for each closed set $\Theta\subset\S^1\setminus\kappa(U_0)$.
\item[(c)] For each closed set $\Theta\subset\S^1\setminus\kappa(U_0)$,
there exists $a>0$ such that
$$
E^{U_0}(\Theta)=E^{V^2}\big([a,\infty)\big)E^{U_0}(\Theta).
$$
\end{enumerate}
\end{Lemma}

One could have the impression that the result of point (c) also holds in the other
direction; namely that for each $a>0$, there exists a closed set
$\Theta\subset\S^1\setminus\kappa(U_0)$ such that
$$
E^{V^2}\big([a,\infty)\big)=E^{U_0}(\Theta)E^{V^2}\big([a,\infty)\big).
$$
But this is not true in general, as can be seen for instance in Example
\ref{ex_Laplacian}.

\begin{proof}
The proof of (a) is similar to the one of \cite[Lemma~2.6(a)]{RT12_0}. (b) follows
directly by invoking a compacity argument. For (c), if $E^{U_0}(\Theta)=0$ or $V^2$
is strictly positive, then the claim is trivial. So, assume that $E^{U_0}(\Theta)\ne0$
and that $V^2$ is not strictly positive, that is, $0\in\sigma(V^2)$. Suppose by absurd
that there is no $a>0$ such that
$E^{U_0}(\Theta)=E^{V^2}\big([a,\infty)\big)E^{U_0}(\Theta)$. Then, for each
$n\in\N^*$, there exists $\psi_n\in\H_0$ such that
$E^{V^2}\big([0,1/n)\big)E^{U_0}(\Theta)\psi_n\neq0$, and the vectors
$$
\varphi_n:=\frac{E^{V^2}\big([0,1/n)\big)E^{U_0}(\Theta)\psi_n}
{\|E^{V^2}\big([0,1/n)\big)E^{U_0}(\Theta)\psi_n\|_{\H_0}}
$$
satisfy $E^{U_0}(\Theta)\varphi_n=E^{V^2}\big([0,1/n)\big)\varphi_n=\varphi_n$ and
$\|\varphi_n\|_{\H_0}=1$. It follows from (b) that
\begin{align*}
{\rm Const.}&\ge\lim_{\varepsilon\searrow0}\big\|\big(V^2\big(V^2+1\big)^{-1}
+\varepsilon\big)^{-1}E^{U_0}(\Theta)\big\|_{\B(\H_0)}\\
&\ge\lim_{\varepsilon\searrow0}\big\|\big(V^2\big(V^2+1\big)^{-1}
+\varepsilon\big)^{-1}E^{U_0}(\Theta)\varphi_n\big\|_{\H_0}\\
&=\lim_{\varepsilon\searrow0}\big\|\big(V^2\big(V^2+1\big)^{-1}
+\varepsilon\big)^{-1}E^{V^2}\big([0,1/n)\big)\varphi_n\big\|_{\H_0}\\
&\ge\lim_{\varepsilon\searrow0}\big(n^{-1}\big(n^{-1}+1\big)^{-1}
+\varepsilon\big)^{-1}\|\varphi_n\|_{\H_0}\\
&=1+n,
\end{align*}
which leads to a contradiction when $n\to\infty$.
\end{proof}

\section{Locally $U_0$-smooth operators}\label{section_smooth}
\setcounter{equation}{0}

In this section, we exhibit a class of locally $U_0$-smooth operators and prove that
$U_0$ has purely absolutely continuous spectrum in $\sigma(U_0)\setminus\kappa(U_0)$
using commutator methods for unitary operators \cite{FRT13}. We start with the
construction of a conjugate operator for $U_0$. For each $j\in\{1,\ldots,d\}$, we set
$\Pi_j:=V_j\big(V_j^2+1\big)^{-1}$. Since
$$
\Pi_j=\big(V_j+i\big)^{-1}+i\big(V_j-i\big)^{-1}\big(V_j+i\big)^{-1},
$$
with $\big(V_j\pm i\big)^{-1}\in C^1(Q_j)$, we have $\Pi_j\in C^1(Q_j)$
\cite[Prop.~5.1.5]{ABG96}. Thus $\Pi_j\dom(Q_j)\subset\dom(Q_j)$, and the operator
$$
\A\;\!\varphi:=\tfrac12\sum_{j=1}^d\big(\Pi_jQ_j+Q_j\Pi_j\big)\varphi,
\quad\varphi\in\cap_{j=1}^d \dom(Q_j),
$$
is well-defined and symmetric. In fact, the operator $\A$ is essentially self-adjoint:

\begin{Lemma}[Conjugate operator for $U_0$]\label{lemma_A}
Let Assumptions \ref{ass_regularity} and \ref{ass_commute} be satisfied.
Then, the operator $\A$ is essentially self-adjoint on $\dom(Q^2)$, and its closure
$A:=\overline{\A}$ is essentially self-adjoint on any core for $Q^2$.
\end{Lemma}

\begin{proof}
We apply the commutator criterion of essential self-adjointness \cite[Thm.~X.37]{RS75}.
Let $\Pi:=(\Pi_1,\ldots,\Pi_d)$, and for $n>1$ define the self-adjoint operator
$N:=Q^2+\Pi^2+n$ with domain $\dom(N)=\dom(Q^2)$. In the form sense on $\dom(N)$,
one has
\begin{align*}
N^2
&=Q^4+\Pi^4+n^2+2n\;\!Q^2+2n\;\!\Pi^2+Q^2\Pi^2+\Pi^2Q^2\\
&=Q^4+\Pi^4+n^2+2n\;\!Q^2+2n\;\!\Pi^2
+\sum_{j,k=1}^d\big(Q_j\Pi_k^2Q_j+\Pi_kQ_j^2\Pi_k\big)+R
\end{align*}
with
$
R:=\sum_{j,k=1}^d\big(\Pi_k[\Pi_k,Q_j]Q_j
+Q_j[Q_j,\Pi_k]\Pi_k+[\Pi_k,Q_j]^2\big)
$.
Now, the following inequality holds
$$
\sum_{j,k=1}^d\big(\Pi_k[\Pi_k,Q_j]Q_j+Q_j[Q_j,\Pi_k]\Pi_k\big)
\ge-dQ^2-\sum_{j,k=1}^d\big|\Pi_k[\Pi_k,Q_j]\big|^2.
$$
Thus, there exists $c>0$ such that $R\ge-dQ^2-c$. Altogether, we have shown
in the form sense on $\dom(N)$ that 
$$
N^2\ge Q^4+\Pi^4+(n^2-c)+(2n-d)\;\!Q^2+2n\;\!\Pi^2
+\sum_{j,k=1}^d\big(Q_j\Pi_k^2Q_j+\Pi_kQ_j^2\Pi_k\big),
$$
where the r.h.s. is a sum of positive terms for $n$ large enough. In particular,
one has for $\varphi\in\dom(N)$ and $j\in\{1,\ldots,d\}$
$$
\|N\varphi\|_{\H_0}^2
\ge\big\|\Pi_jQ_j\varphi\big\|_{\H_0}^2+\big\|Q_j\Pi_j\varphi\big\|_{\H_0}^2,
$$
which implies that
$$
\|A\;\!\varphi\|_{\H_0}\le\tfrac12\sum_{j=1}^d\left(\big\|\Pi_jQ_j\varphi\big\|_{\H_0}
+\big\|\Phi_j\Pi_j\varphi\big\|_{\H_0}\right)
\le d\;\!\|N\varphi\|_{\H_0}.
$$

It remains to estimate the commutator $[A,N]$. In the form sense on $\dom(N)$, one has
\begin{align*}
2\;\![A,N]
&=\sum_{j,k=1}^d\big([\Pi_j,Q_k]Q_jQ_k+Q_k[\Pi_j,Q_k]Q_j
+Q_j[\Pi_j,Q_k]Q_k+Q_jQ_k[\Pi_j,Q_k]\\
&\quad+\Pi_j[Q_j,\Pi_k]\Pi_k+\Pi_j\Pi_k[Q_j,\Pi_k]+[Q_j,\Pi_k]\Pi_j\Pi_k
+\Pi_k[Q_j,\Pi_k]\Pi_j\big).
\end{align*}
The last four terms are bounded. For the other terms, the fact that $\Pi_j\in C^2(Q)$,
together with the bounds
$$
\left|\big\langle Q_j\varphi,BQ_k\varphi\big\rangle_{\H_0}\right|
\le\|B\|_{\B(\H_0)}\big\langle\varphi,Q^2\varphi\big\rangle_{\H_0}
\le{\rm Const.}\;\!\langle\varphi,N\varphi\rangle_{\H_0},
\quad\varphi\in\dom(N),~B\in\B(\H),
$$
leads to the desired estimate, namely,
$
\big\langle\varphi,[A,N]\varphi\big\rangle_{\H_0}
\le{\rm Const.}\;\!\langle\varphi,N\varphi\rangle_{\H_0}
$.
\end{proof}

The operator $U_0$ is regular with respect to $A:$

\begin{Lemma}\label{lemma_C2}
Let Assumptions \ref{ass_regularity} and \ref{ass_commute} be satisfied. Then, the
operator $U_0$ is of class $C^2(A)$ with
$[A,U_0]=U_0\sum_{j=1}^dV_j^2\big(V_j^2+1\big)^{-1}$.
\end{Lemma}

\begin{proof}
Let $\varphi\in\D\subset\dom(Q^2)$. Using the fact that $[U_0,\Pi_j]=0$ for each
$j\in\{1,\ldots,d\}$, we obtain
\begin{align*}
&\langle A\;\!\varphi,U_0\varphi\rangle_{\H_0}
-\langle\varphi,U_0A\;\!\varphi\rangle_{\H_0}\\
&=\tfrac12\sum_{j=1}^d
\left(\big\langle\big(\Pi_jQ_j+Q_j\Pi_j\big)\varphi,U_0\varphi\big\rangle_{\H_0}
-\big\langle\varphi,U_0\big(\Pi_jQ_j+Q_j\Pi_j\big)\varphi\big\rangle_{\H_0}\right)\\
&=\tfrac12\sum_{j=1}^d\tfrac\d{\d x_j}\Big|_{x=0}
\left(\big\langle-i\big(\Pi_j\e^{ix\cdot Q}+\e^{ix\cdot Q}\Pi_j\big)\varphi,
U_0\varphi\big\rangle_{\H_0}-\big\langle\varphi,iU_0\big(\Pi_j\e^{-ix\cdot Q}
+\e^{-ix\cdot Q}\Pi_j\big)\varphi\big\rangle_{\H_0}\right)\\
&=\tfrac12\sum_{j=1}^d\tfrac\d{\d x_j}\Big|_{x=0}\left(\big\langle\varphi,i\;\!\Pi_j
\big[\e^{-ix\cdot Q},U_0\big]\varphi\big\rangle_{\H_0}
+\big\langle\varphi,i\big[\e^{-ix\cdot Q},U_0\big]\Pi_j\;\!\varphi
\big\rangle_{\H_0}\right).
\end{align*}
Now, a direct calculation shows that
$$
\big[\e^{-ix\cdot Q},U_0\big]
=\big(U_0(x)U_0^{-1}-1\big)U_0\e^{-ix\cdot Q}
=\e^{-ix\cdot Q}U_0\big(1-U_0(-x)U_0^{-1}\big),
$$
and the fact that $\Pi_j\in C^1(Q_j)$ implies that $\Pi_j\;\!\varphi\in\dom(Q_j)$.
Thus,
\begin{align*}
&\langle A\;\!\varphi,U_0\varphi\rangle_{\H_0}
-\langle\varphi,U_0A\;\!\varphi\rangle_{\H_0}\\
&=\tfrac12\sum_{j=1}^d\tfrac\d{\d x_j}\Big|_{x=0}
\left(\big\langle\e^{ix\cdot Q}\Pi_j\;\!\varphi,iU_0\big(1-U_0(-x)U_0^{-1}\big)
\varphi\big\rangle_{\H_0}+\big\langle\big(U_0(x)U_0^{-1}-1\big)^*\varphi,
iU_0\e^{-ix\cdot Q}\Pi_j\;\!\varphi\big\rangle_{\H_0}\right)\\
&=\tfrac12\sum_{j=1}^d
\left(\big\langle\Pi_j\;\!\varphi,U_0V_j\;\!\varphi\big\rangle_{\H_0}
+\big\langle V_j\;\!\varphi,U_0\Pi_j\;\!\varphi\big\rangle_{\H_0}\right)\\
&=\left\langle\varphi,U_0\sum_{j=1}^dV_j^2\big(V_j^2+1\big)^{-1}\varphi
\right\rangle_{\H_0}.
\end{align*}
Since $\D$ is a core for $Q^2$, and thus for $A$ by Lemma \ref{lemma_A}, this implies
that $U_0\in C^1(A)$ with
$$
[A,U_0]=U_0\sum_{j=1}^dV_j^2\big(V_j^2+1\big)^{-1}.
$$
Finally, since $U_0\in C^1(A)$ and
$\sum_{j=1}^dV_j^2\big(V_j^2+1\big)^{-1}\in C^1(A)$, we infer from
\cite[Prop.~5.1.5]{ABG96} that
$$
[A,U_0]=U_0\sum_{j=1}^dV_j^2\big(V_j^2+1\big)^{-1}\in C^1(A),
$$
and thus that $U_0\in C^2(A)$.
\end{proof}

Using Lemma \ref{lemma_C2}, we can prove a Mourre estimate for $U_0$ on the set
$\S^1\setminus\kappa(U_0):$

\begin{Lemma}[Mourre estimate for $U_0$]\label{lemma_Mourre}
Let Assumptions \ref{ass_regularity} and \ref{ass_commute} be satisfied, and let
$\e^{i\lambda}\in\S^1\setminus\kappa(U_0)$. Then, there exist $a,\delta>0$ such that
$$
E^{U_0}(\lambda;\delta)U_0^{-1}[A,U_0]E^{U_0}(\lambda;\delta)
\ge a\;\!E^{U_0}(\lambda;\delta).
$$
\end{Lemma}

\begin{proof}
Since $\e^{i\lambda}\in\S^1\setminus\kappa(U_0)$ and $U_0$ and $V^2$ strongly commute,
there exists $\delta>0$ such that
\begin{align}
{\rm Const.}&\ge\lim_{\varepsilon\searrow0}\big\|\big(V^2\big(V^2+1\big)^{-1}
+\varepsilon\big)^{-1}E^{U_0}(\lambda;\delta)\big\|_{\B(\H_0)}\nonumber\\
&=\lim_{\varepsilon\searrow0}\big\|E^{U_0}(\lambda;\delta)
\big(V^2\big(V^2+1\big)^{-1}E^{U_0}(\lambda;\delta)
+\varepsilon\big)^{-1}E^{U_0}(\lambda;\delta)\big\|_{\B(\H_0)}\nonumber\\
&=\lim_{\varepsilon\searrow0}\big\|\big(V^2\big(V^2+1\big)^{-1}E^{U_0}(\lambda;\delta)
+\varepsilon\big)^{-1}\big\|_{\B(\H_{\lambda,\delta})},\label{eq_bound}
\end{align}
with $\H_{\lambda,\delta}:=E^{U_0}(\lambda;\delta)\H_0$. Furthermore, we have
$$
\big\|\big(V^2\big(V^2+1\big)^{-1}E^{U_0}(\lambda;\delta)
+\varepsilon\big)^{-1}\big\|_{\B(\H_{\lambda,\delta})}
=(a+\varepsilon)^{-1},
$$
with $a\ge0$ the infimum of the spectrum of
$V^2\big(V^2+1\big)^{-1}E^{U_0}(\lambda;\delta)$ in $\H_{\lambda,\delta}$. Thus,
\eqref{eq_bound} entails the bound $a^{-1}\le{\rm Const.}$, which implies that
$a\ne0$. In consequence,
$$
V^2\big(V^2+1\big)^{-1}E^{U_0}(\lambda;\delta)\ge a\;\!E^{U_0}(\lambda;\delta),
$$
with $a>0$. This fact, together with the equality
$[A,U_0]=U_0\sum_{j=1}^dV_j^2\big(V_j^2+1\big)^{-1}$ of Lemma \ref{lemma_C2},
implies that
\begin{align*}
E^{U_0}(\lambda;\delta)U_0^{-1}[A,U_0]E^{U_0}(\lambda;\delta)
&=\sum_{j=1}^dV_j^2\big(V_j^2+1\big)^{-1}E^{U_0}(\lambda;\delta)\\
&\ge V^2\big(V^2+1\big)^{-1}E^{U_0}(\lambda;\delta)\\
&\ge a\;\!E^{U_0}(\lambda;\delta),
\end{align*}
which proves the claim.
\end{proof}

We now exhibit a class of locally $U_0$-smooth operators and prove that $U_0$ has
purely absolutely continuous spectrum in $\sigma(U_0)\setminus\kappa(U_0)$. For this,
we recall that an operator $B\in\B(\H_0)$ is locally $U_0$-smooth on an open set
$\Theta\subset\S^1$ if for each closed set $\Theta'\subset\Theta$ there exists
$c_{\Theta'}\ge0$ such that
$$
\sum_{n\in\Z}\big\|B\;\!U_0^nE^{U_0}(\Theta')\varphi\big\|_{\B(\H_0)}^2
\le c_{\Theta'}\|\varphi\|_{\H_0}^2
\quad\hbox{for all $\varphi\in\H_0$.}
$$
We also recall that the space $\big(\dom(A),\H_0\big)_{1/2,1}$ is defined by real
interpolation (see \cite[Sec.~3.4.1]{ABG96}). Since $\Pi_j\in C^1(Q_j)$ for each
$j\in\{1,\ldots,d\}$, we have $\dom(\langle Q\rangle)\subset\dom(A)$. Therefore, it
follows by interpolation \cite[Thm.~2.6.3 \& Thm.~3.4.3.(a)]{ABG96} that we have the
continuous embeddings
\begin{equation}\label{eq_embeddings}
\dom(\langle Q\rangle^s)
\subset\big(\dom(A),\H_0\big)_{1/2,1}
\subset\H
\subset\left(\big(\dom(A),\H_0\big)_{1/2,1}\right)^*
\subset\dom(\langle Q\rangle^{-s}),
\quad s>1/2.
\end{equation}

\begin{Theorem}[Locally $U_0$-smooth operators]\label{thm_spectrum}
Let Assumptions \ref{ass_regularity} and \ref{ass_commute} be satisfied.
\begin{enumerate}
\item[(a)] The spectrum of $U_0$ in $\sigma(U_0)\setminus\kappa(U_0)$ is purely
absolutely continuous.
\item[(b)] Each operator $B\in\B\big(\dom(\langle Q\rangle^{-s}),\H_0\big)$, with
$s>1/2$, is locally $U_0$-smooth on $\S^1\setminus\kappa(U_0)$.
\end{enumerate}
\end{Theorem}

\begin{proof}
The first claim follows from Lemmas \ref{lemma_C2}-\ref{lemma_Mourre} and
\cite[Thm.~2.7]{FRT13}. The second claim follows from the embeddings
\eqref{eq_embeddings} and \cite[Prop~2.9]{FRT13}.
\end{proof}

\begin{Example}[$V$ constant]\label{ex_V_constant}
Assume that there exist a dense set $\G\subset\H_0$ and $v\in\R^d\setminus\{0\}$ such
that
$$
\emph{s\hspace{1pt}-}\tfrac\d{\d t}\;\!iU_0(te_j)U_0^{-1}\varphi
=v_j\;\!U_0(te_j)U_0^{-1}\varphi,\quad t\in\R,~j\in\{1,\ldots,d\},~\varphi\in\G.
$$
Then, $U_0(te_j)=\e^{-itv_j}U_0$ on $\G$, and thus $U_0(te_j)=\e^{-itv_j}U_0$ by the
density of $\G$. It follows that $U_0(x)=\e^{-ix\cdot v}U_0$ for each $x\in\R^d$ and
that $V=v$. Thus, Assumptions \ref{ass_regularity} and \ref{ass_commute} are
satisfied. Moreover, since we have for all $\e^{i\lambda}\in\S^1$ and $\delta>0$ that
$$
\lim_{\varepsilon\searrow0}\big\|\big(V^2\big(V^2+1\big)^{-1}+\varepsilon\big)^{-1}
E^{U_0}(\lambda;\delta)\big\|_{\B(\H_0)}
\le v^{-2}\;\!\big(v^2+1\big)
<\infty,
$$
the set $\kappa(U_0)$ of critical values of $U_0$ is empty, and Theorem
\ref{thm_spectrum}(a) implies that $\sigma(U_0)=\sigma_{\rm ac}(U_0)$.
\end{Example}

\begin{Example}[Time-one propagator for the Laplacian]\label{ex_Laplacian}
Let $U_0$ be the time-one propagator for the Laplacian in $\R^d$. That is, let
$Q=(Q_1,\ldots,Q_d)$ and $P=(P_1,\ldots,P_d)$ be the usual families of position and
momentum operators in the Hilbert space $\H_0:=\ltwo(\R^d)$, and let $U_0=\e^{-iP^2}$
be the time-one propagator for the operator $P^2$ in $\H_0$. In such a case, the set
of Schwartz functions $\D:=\SS(\R^d)$ is a core for $Q^2$, we have for each $x\in\R^d$
$$
U_0(x)
=\e^{-ix\cdot Q}\e^{-iP^2}\e^{ix\cdot Q}
=\e^{-i(P+x)^2},
$$
and the operator $U_0(x)U_0^{-1}=\e^{-i(P+x)^2}\e^{iP^2}$ is strongly differentiable
on $\D$. Straightforward calculations show that $V_j=2P_j$, $V'_{jk}=2\delta_{jk}$ and
$V''_{jk\ell}=0$ for $j,k,\ell\in\{1,\ldots,d\}$ and $\delta_{jk}$ the Kronecker delta
function. Thus, Assumptions \ref{ass_regularity} and \ref{ass_commute} are satisfied.
Moreover, since we have for $\e^{i\lambda}\in\S^1$ and $\delta>0$ that
$$
\lim_{\varepsilon\searrow0}\left\|\big(V^2\big(V^2+1\big)^{-1}+\varepsilon\big)^{-1}
E^{U_0}(\lambda;\delta)\right\|_{\B(\H_0)}
=\lim_{\varepsilon\searrow0}\left\|\big(4P^2\big(4P^2+1\big)^{-1}+\varepsilon\big)^{-1}
E^{\e^{-iP^2}}(\lambda;\delta)\right\|_{\B(\H_0)},
$$
the set of critical values of $U_0$ is the singleton $\kappa(U_0)=\{1\}$, and Theorem
\ref{thm_spectrum}(a) implies that $\sigma(U_0)=\sigma_{\rm ac}(U_0)$.
\end{Example}

\section{Summation formula}\label{section_summation}
\setcounter{equation}{0}

In this section, we prove and give an interpretation of a summation formula which
relates the evolution of the localisation operator $f(Q)$ under $U_0$ to a time
operator $T_f$.

\subsection{Averaged localisation functions}\label{section_averaged}

First, we recall some properties of a class of averaged localisation functions which
appears naturally when dealing with quantum time delay. These functions, which are
denoted $R_f$, are constructed in terms of functions $f\in\linf(\R^d)$ of localisation
around the origin $0$ of $\R^d$. They were already used, in one form or another, in
\cite{GT07,RT12_0,RT12,Tie08,Tie09}.

\begin{Assumption}\label{ass_f}
The function $f\in\linf(\R^d)$ satisfies the following conditions:
\begin{enumerate}
\item[(i)] There exists $\rho>0$ such that
$|f(x)|\le{\rm Const.}\;\!\langle x\rangle^{-\rho}$ for almost every $x\in\R^d$.
\item[(ii)] $f=1$ on a neighbourhood of~~$0$.
\end{enumerate}
\end{Assumption}

If $f$ satisfies Assumption \ref{ass_f}, then $\slim_{r\to\infty}f(Q/r)=1$.
Furthermore, one has for each $x\in\R^d\setminus\{0\}$
$$
\left|\int_0^\infty\frac{\d\mu}\mu\big(f(\mu x)-\chi_{[0,1]}(\mu)\big)\right|
\le\int_0^1\frac{\d\mu}\mu\,|f(\mu x)-1|
+{\rm Const.}\int_1^{+\infty}\d\mu\,\mu^{-(1+\rho)}
<\infty,
$$
where $\chi_{[0,1]}$ denotes the characteristic function for the interval $[0,1]$.
Therefore the function
$$
R_f:\R^d\setminus\{0\}\to\C,
~~x\mapsto\int_0^{+\infty}\frac{\d\mu}\mu\big(f(\mu x)-\chi_{[0,1]}(\mu)\big)
$$
is well-defined. If $\R^*_+:=(0,\infty)$, endowed with the multiplication, is seen as
a Lie group with Haar measure $\frac{\d\mu}\mu$, then $R_f$ is the renormalised
average of $f$ with respect to the (dilation) action of $\R^*_+$ on $\R^d$.

In the next lemma we recall some differentiability and homogeneity properties of
$R_f$. We also give the explicit form of $\nabla R_f$ when $f$ is a radial function.
The reader is referred to \cite[Sec.~2]{Tie09} for proofs and details.

\begin{Lemma}\label{lemma_R_f}
Let Assumption \ref{ass_f} be satisfied.
\begin{enumerate}
\item[(a)] If $(\partial_jf)(x)$ exists for all $j\in\{1,\ldots,d\}$ and $x\in\R^d$,
and if there exists some $\rho'>0$ such that
$|(\partial_jf)(x)|\le{\rm Const.}\langle x\rangle^{-(1+\rho')}$ for each $x\in\R^d$,
then $R_f$ is differentiable on $\R^d\setminus\{0\}$ with partial derivative given by
$$
\big(\partial_jR_f\big)(x)=\int_0^\infty\d\mu\,(\partial_jf)(\mu x).
$$
In particular, if $f\in\SS(\R^d)$ then $R_f\in C^\infty(\R^d\setminus\{0\})$.
\item[(b)] If $R_f\in C^m(\R^d\setminus\{0\})$ for some $m\ge1$, then $R_f$ satisfies
the homogeneity properties
\begin{gather}
x\cdot(\nabla R_f)(x)=-1,\label{eq_homo_1}\\
t^{|\alpha|}(\partial^\alpha R_f)(tx)=(\partial^\alpha R_f)(x),\label{eq_homo_2}
\end{gather}
where $x\in\R^d\setminus\{0\}$, $t>0$ and $\alpha\in\N^d$ is a multi-index with
$1\le|\alpha|\le m$.
\item[(c)] If $f$ is radial, i.e. there exists $f_0\in\linf(\R)$ such that
$f(x)=f_0(|x|)$ for almost every $x\in\R^d$, then
$R_f\in C^\infty(\R^d\setminus\{0\})$ and $(\nabla R_f)(x)=-x^{-2}x$ for
$x\in\R^d\setminus\{0\}$.
\end{enumerate}
\end{Lemma}

Obviously, one can show as in Lemma \ref{lemma_R_f}(a) that
$R_f\in C^m(\R^d\setminus\{0\})$ if $(\partial^\alpha f)(x)$ exists for all $x\in\R^d$
and $\alpha\in\N^d$ with $|\alpha|\le m$, and
$|(\partial^\alpha f)(x)|\le{\rm Const.}\langle x\rangle^{-(|\alpha|+\rho')}$ some
$\rho'>0$. However, this is not a necessary condition. In some cases (as in Lemma 
\ref{lemma_R_f}(c)), the function $R_f$ is very regular outside the point $0$ even if
$f$ is not continuous.

\subsection{Proof of the summation formula}\label{section_proof}

In the sequel, we let $D$ be any self-adjoint operator in $\H_0$ satisfying the
following: $D$ and $U_0$ strongly commute, and if $\rho\in C^\infty_{\rm c}(\R)$, then
there exists $\eta\in C^\infty_{\rm c}(\R)$ such that
$\rho(D)=\rho(D)\eta\big(V^2\big)$. Obviously, the simplest choice is to take $D=V^2$,
but in certain cases other choices can be more convenient. For instance, when $U_0$ is
the time-one propagator of some self-adjoint operator $H_0$, that is,
$U_0=\e^{-iH_0}$, it can be more advantageous to take $D=H_0$ (see Section
\ref{section_symmetrised} for more comments on this). With the operator $D$ at hand,
we define for each $t\ge0$ the set
$$
\dom_t:=\left\{\varphi\in\dom(\langle Q\rangle^t)
\mid\varphi=\zeta(U_0)\varphi=\rho(D)\varphi
\hbox{ for some $\zeta\in C^\infty_{\rm c}\big(\S^1\setminus\kappa(U_0)\big)$ and
$\rho\in C^\infty_{\rm c}(\R)$}\right\}.
$$
The sets $\dom_t$ are well-defined because the set of critital values $\kappa(U_0)$ is
closed due to Lemma \ref{lemma_kappa}(a). Furthermore, we have
$\dom_{t_1}\subset\dom_{t_2}$ if $t_1\ge t_2$, and Theorem \ref{thm_spectrum}(a)
implies that $\dom_t$ is included in the subspace $\H_{\rm ac}(U_0)$ of absolute
continuity of $U_0$.

In the next proposition, we define the operator $T_f$. For that purpose, we consider
the operators $V'_{jk}$ as the components of a $d$-dimensional (Hessian) matrix which
we denote by $V'$ ($V'^\intercal$ stands for its matrix transpose). Also, we use
sometimes the notation $C^{-1}$ for an operator $C$ a priori not invertible. In such a
case, the operator $C^{-1}$ is restricted to a set where it is well-defined.

\begin{Proposition}[Operator $T_f$]\label{prop_T_f}
Let Assumptions \ref{ass_regularity}, \ref{ass_commute} and \ref{ass_f} be satisfied,
and assume that $R_f\in C^1(\R^d\setminus\{0\})$. Then, the map
$$
t_f:\dom_1\to\C,~~\varphi\mapsto
\tfrac12\sum_{j=1}^d\left(\big\langle Q_j\varphi,
\big(\partial_jR_f\big)(V)\varphi\big\rangle_{\H_0}
+\big\langle\big(\partial_jR_{\overline f}\big)(V)\varphi,
Q_j\varphi\big\rangle_{\H_0}\right),
$$
is well-defined. Moreover, if $\big(\partial_jR_f\big)(V)\varphi\in\dom(Q_j)$ for each
$j\in\{1,\ldots,d\}$, then the operator
\begin{equation}\label{eq_T_f}
T_f\;\!\varphi:=\tfrac12\left(Q\cdot(\nabla R_f)(V)
+(\nabla R_f)\big(\tfrac V{|V|}\big)\cdot Q\;\!|V|^{-1}
+i(\nabla R_f)\big(\tfrac V{|V|}\big)\cdot
\big(V'^\intercal V\big)|V|^{-3}\right)\varphi,\quad\varphi\in\dom_1,
\end{equation}
satisfies $t_f(\varphi)=\big\langle\varphi,T_f\;\!\varphi\big\rangle_{\H_0}$ for each
$\varphi\in\dom_1$. In particular, $T_f$ is a symmetric operator if $f$ is real and
$\dom_1$ is dense in $\H_0$.
\end{Proposition}

\begin{Remark}\label{remark_T_f}
(a) The operator on the r.h.s. of \eqref{eq_T_f} is rather complicated, and one could
be tempted to replace it by the simpler operator
$\frac12\big(Q\cdot(\nabla R_f)(V)+(\nabla R_f)(V)\cdot Q\big)$. Unfortunately, a
precise meaning for this operator is not available at this level of generality; it can
be rigorously defined only in concrete examples.

(b) If $\varphi\in\dom_1$ and $f$ either belongs to $\SS(\R^d)$ or is radial, then the
assumption $\big(\partial_jR_f\big)(V)\varphi\in\dom(Q_j)$ holds for each
$j\in\{1,\ldots,d\}$. Indeed, due to Lemma \ref{lemma_kappa}(c) and the definition of
$\dom_1$, there exists $\eta\in C^\infty_{\rm c}\big((0,\infty)\big)$ such that
$
\big(\partial_jR_f\big)(V)\varphi
=\big(\partial_jR_f\big)(V)\eta\big(V^2\big)\varphi
$,
and we have the inclusion $\big(\partial_jR_f\big)(V)\eta\big(V^2\big)\in C^1(Q_j)$
due to Lemma \ref{lemma_R_f} and \cite[Prop.~5.1]{RT12_0}. Thus,
$\big(\partial_jR_f\big)(V)\varphi\in\dom(Q_j)$.
\end{Remark}

\begin{proof}[Proof of Proposition \ref{prop_T_f}]
Let $\varphi\in\dom_1$. Then, there exists
$\eta\in C^\infty_{\rm c}\big((0,\infty)\big)$ such that
$
\big(\partial_jR_f\big)(V)\varphi
=\big(\partial_jR_f\big)(V)\eta\big(V^2\big)\varphi
$.
Thus, we have
$
\big\|\big(\partial_jR_f\big)(V)\varphi\big\|_{\H_0}
\le{\rm Const.}\;\!\|\varphi\|_{\H_0}
$
as in Remark \ref{remark_T_f}(b), and we obtain
$$
\big|t_f(\varphi)\big|
\le{\rm Const.}\;\!\|\varphi\|_{\H_0}\;\!\|\langle Q\rangle\varphi\|_{\H_0},
$$
which implies the first part of the proposition. For the second part, it is sufficient
to show that
$$
\sum_{j=1}^d\big\langle\big(\partial_jR_{\overline f}\big)(V)\varphi,
Q_j\varphi\big\rangle_{\H_0}
=\left\langle\varphi,\left((\nabla R_f)\big(\tfrac V{|V|}\big)
\cdot Q\;\!|V|^{-1}+i(\nabla R_f)\big(\tfrac V{|V|}\big)
\cdot\big(V'^\intercal V\big)|V|^{-3}\right)\varphi\right\rangle_{\H_0}.
$$
Using Formula \eqref{eq_homo_2} and \cite[Eq.~4.3.2]{Dav95}, we get
\begin{align}
&\sum_{j=1}^d\big\langle\big(\partial_jR_{\overline f}\big)(V)\varphi,
Q_j\varphi\big\rangle_{\H_0}\nonumber\\
&=\sum_{j=1}^d\left\langle(\partial_jR_{\overline f})\big(\tfrac V{|V|}\big)
|V|^{-1}\varphi,Q_j\varphi\right\rangle_{\H_0}\nonumber\\
&=\sum_{j=1}^d\lim_{\varepsilon\searrow0}\;\!\lim_{\delta\to0}
\left\langle\big(\partial_jR_{\overline f}\big)\big(\tfrac V{|V|}\big)\varphi,
\big(V^2+\varepsilon\big)^{-1/2}Q_j\big(1+i\delta Q_j\big)^{-1}\varphi
\right\rangle_{\H_0}\nonumber\\
&=\left\langle\varphi,(\nabla R_f)\big(\tfrac V{|V|}\big)\cdot Q\;\!|V|^{-1}
\varphi\right\rangle_{\H_0}\nonumber\\
&\quad+\pi^{-1}\sum_{j=1}^d\lim_{\varepsilon\searrow0}\;\!\lim_{\delta\to0}
\int_0^\infty\d t\,t^{-1/2}\left\langle
\big(\partial_jR_{\overline f}\big)\big(\tfrac V{|V|}\big)\varphi,
\left[\big(V^2+\varepsilon+t\big)^{-1},Q_j\big(1+i\delta Q_j\big)^{-1}\right]
\varphi\right\rangle_{\H_0}.\label{eq_Davies}
\end{align}
Now, using Assumption \ref{ass_regularity} and the fact that
$\varphi=\eta\big(V^2\big)\varphi$ with
$\eta\in C^\infty_{\rm c}\big((0,\infty)\big)$, we obtain that
$$
\lim_{\delta\to0}\left[\big(V^2+\varepsilon+t\big)^{-1},
Q_j\big(1+i\delta Q_j\big)^{-1}\right]\varphi
=2i\big(V^2+\varepsilon+t\big)^{-2}\big(V'^\intercal V\big)_j\;\!\varphi.
$$
So,
\begin{align*}
\eqref{eq_Davies}
&=\pi^{-1}\sum_{j=1}^d\lim_{\varepsilon\searrow0}\int_0^\infty\d t\,t^{-1/2}
\left\langle\big(\partial_jR_{\overline f}\big)\big(\tfrac V{|V|}\big)\varphi,
2i\big(V^2+\varepsilon+t\big)^{-2}\big(V'^\intercal V\big)_j
\varphi\right\rangle_{\H_0}\\
&=\sum_{j=1}^d\lim_{\varepsilon\searrow0}
\left\langle\big(\partial_jR_{\overline f}\big)\big(\tfrac V{|V|}\big)\varphi,
i\big(V^2+\varepsilon\big)^{-3/2}\big(V'^\intercal V\big)_j
\varphi\right\rangle_{\H_0}\\
&=\left\langle\varphi,i(\nabla R_f)\big(\tfrac V{|V|}\big)
\cdot\big(V'^\intercal V\big)\;\!|V|^{-3}\varphi\right\rangle_{\H_0},
\end{align*}
and thus
$$
\sum_{j=1}^d\big\langle\big(\partial_jR_{\overline f}\big)(V)\varphi,
Q_j\varphi\big\rangle_{\H_0}
=\left\langle\varphi,\left((\nabla R_f)\big(\tfrac V{|V|}\big)
\cdot Q\;\!|V|^{-1}+i(\nabla R_f)\big(\tfrac V{|V|}\big)
\cdot\big(V'^\intercal V\big)|V|^{-3}\right)\varphi\right\rangle_{\H_0}.
$$
\end{proof}

If $f$ is radial, then $\big(\partial_jR_f\big)(x)=-x^{-2}x_j$ due to Lemma
\ref{lemma_R_f}(c), and Formula \eqref{eq_T_f} holds by Remark \ref{remark_T_f}(b).
Thus,
\begin{equation}\label{eq_T}
T_f=T:=-\tfrac12\left(Q\cdot\tfrac V{V^2}+\tfrac V{|V|}\cdot Q\;\!|V|^{-1}
+\tfrac{i\;\!V}{V^4}\cdot\big(V'^\intercal V\big)\right)
\quad\hbox{on}\quad\dom_1.
\end{equation}

In the next lemma, we establish identities necessary for the proof of the main theorem
of this section. We use the symbol $\F$ for the Fourier transformation on $\R^d$, and
the symbol $\underline\d x$ for the measure on $\R^d$ making $\F$ a unitary operator
in $\ltwo(\R^d,\underline\d x)$.

\begin{Lemma}\label{lemma_trotter}
Let Assumptions \ref{ass_regularity} and \ref{ass_commute} be satisfied.
\begin{enumerate}
\item[(a)] For each compact set $I\subset\R$, $f\in\SS(\R^d)$, $n\in\Z$ and
$\nu>0$, we have the identities
\begin{align}
U_0^{-n}f(\nu Q)U_0^nE^{V^2}(I)
&=\int_{\R^d}\underline\d x\,(\F f)(x)\e^{i\nu x\cdot Q}
\e^{in\int_0^\nu\d s\,(x\cdot V(sx))E^{V^2}(I)}E^{V^2}(I),\label{eq_Trotter_1}\\
E^{V^2}(I)U_0^nf(\nu Q)U_0^{-n}
&=E^{V^2}(I)\int_{\R^d}\underline\d x\,(\F f)(x)
\e^{-in\int_{-\nu}^0\d s\,(x\cdot V(sx))E^{V^2}(I)}
\e^{i\nu x\cdot Q},\label{eq_Trotter_2}
\end{align}
with $\int_0^\nu\d s\,\big(x\cdot V(sx)\big)E^{V^2}(I)$ the Bochner integral of the map
$$
[0,\nu]\ni s\mapsto\big(x\cdot V(sx)\big)E^{V^2}(I)\in\B(\H_0).
$$
\item[(b)] For each compact set $I\subset\R$, $k\in\{1,\ldots,d\}$, $x\in\R^d$ and
$\nu>0$, we have the identity
$$
\frac\d{\d x_k}\int_0^\nu\d s\,\big(x\cdot V(sx)\big)E^{V^2}(I)
=\nu\;\!V_k(\nu x)E^{V^2}(I).
$$
with $\frac\d{\d x_k}$ the derivative in the topology of $\B(\H_0)$.
\end{enumerate}
\end{Lemma}

\begin{proof}
(a) Using functional calculus, we obtain
$$
U_0^{-n}f(\nu Q)U_0^n
=f\big(\nu U_0^{-n}QU_0^n\big)
=\int_{\R^d}\underline\d x\,(\F f)(x)\e^{i\nu x\cdot(U_0^{-n}QU_0^n)}
=\int_{\R^d}\underline\d x\,(\F f)(x)\e^{i\nu U_0^{-n}(x\cdot Q)U_0^n}.
$$
Moreover, we know from Lemma \ref{lemma_V}(d) that
$$
U_0^{-n}(x\cdot Q)U_0^n
=\overline{\big(x\cdot Q+n\;\!(x\cdot V)\big)
\upharpoonright\dom(x\cdot Q)\cap\dom(x\cdot V)}.
$$
Thus, it follows by the Trotter-Kato formula \cite[Thm.~VIII.31]{RS80} that
$$
U_0^{-n}f(\nu Q)U_0^{-n}
=\int_{\R^d}\underline\d x\,(\F f)(x)\slim_{m\to\infty}
\left(\e^{i\nu(x\cdot Q)/m}\e^{i\nu n(x\cdot V)/m}\right)^m.
$$
Now, an induction argument shows that
$
\left(\e^{i\nu(x\cdot Q)/m}\e^{i\nu n(x\cdot V)/m}\right)^m
=\e^{i\nu(x\cdot Q)}\e^{\frac{i\nu n}m\sum_{\ell=1}^m
x\cdot V\big(\frac{(\ell-1)\nu}m\;\!x\big)}
$.
Indeed, for $m=1$ the claim is trivial. For $m-1\ge1$, we assume that the claim is
true. Then, for $m\ge3$ the change of variable $y:=\frac{m-1}m\;\!x$ and the induction
hypothesis imply that
\begin{align*}
\left(\e^{i\nu(x\cdot Q)/m}\e^{i\nu n(x\cdot V)/m}\right)^m
&=\left(\e^{i\nu(y\cdot Q)/(m-1)}\e^{i\nu n(y\cdot V)/(m-1)}\right)^{m-1}
\e^{i\nu(x\cdot Q)/m}\e^{i\nu n(x\cdot V)/m}\\
&=\e^{i\nu(y\cdot Q)}\e^{\frac{i\nu n}{m-1}\sum_{\ell=1}^{m-1}
y\cdot V\big(\frac{(\ell-1)\nu}{m-1}y\big)}
\e^{i\nu(x\cdot Q)/m}\e^{i\nu n(x\cdot V)/m}\\
&=\e^{i\nu(m-1)(x\cdot Q)/m}\e^{\frac{i\nu n}m\sum_{\ell=2}^m
x\cdot V\big(\frac{(\ell-2)\nu}mx\big)}
\e^{i\nu(x\cdot Q)/m}\e^{i\nu n(x\cdot V)/m}\\
&=\e^{i\nu(x\cdot Q)}\e^{\frac{i\nu n}m\sum_{\ell=2}^m
x\cdot V\big(\frac{(\ell-2)\nu}mx+\frac\nu m\;\!x\big)}\e^{i\nu n(x\cdot V)/m}\\
&=\e^{i\nu(x\cdot Q)}\e^{\frac{i\nu n}m\sum_{\ell=1}^m
x\cdot V\big(\frac{(\ell-1)\nu}mx\big)}.
\end{align*}
Thus,
\begin{align*}
U_0^{-n}f(\nu Q)U_0^nE^{V^2}(I)
=\int_{\R^d}\underline\d x\,(\F f)(x)\e^{i\nu(x\cdot Q)}\slim_{m\to\infty}
\e^{\frac{i\nu n}m\sum_{\ell=1}^mx\cdot V\big(\frac{(\ell-1)\nu}mx\big)}E^{V^2}(I).
\end{align*}
But, using the continuity of the map $\B(\H_0)\ni B\mapsto\e^B\in\B(\H_0)$, the mutual
strong commutation of the operators $V_j(\;\!\cdot\;\!)$ and the boundedness of the
operator $\big(x\cdot V\big(\frac{(\ell-1)\nu}mx\big)\big)E^{V^2}(I)$, we obtain that
\begin{align*}
\slim_{m\to\infty}\e^{\frac{i\nu n}m\sum_{\ell=1}^mx\cdot V
\big(\frac{(\ell-1)\nu}mx\big)}E^{V^2}(I)
&=\e^{\ulim_{m\to\infty}\frac{i\nu n}m\sum_{\ell=1}^m
\big(x\cdot V\big(\frac{(\ell-1)\nu}mx\big)\big)E^{V^2}(I)}E^{V^2}(I)\\
&=\e^{in\int_0^\nu\d s\,(x\cdot V(sx))E^{V^2}(I)}E^{V^2}(I),
\end{align*}
with $\ulim$ the limit in the topology of $\B(\H_0)$. This concludes the proof of
\eqref{eq_Trotter_1}. The proof of \eqref{eq_Trotter_2} is similar.

(b) Let $j\in\{1,\ldots,d\}$. Then, Assumption \ref{ass_regularity} and
\cite[Prop.~5.1.2(b)]{ABG96} imply that $V_j\in C^2\big(Q,\dom(V_j),\H_0\big)$, which
in turns implies that $V_j\in C^1_{\rm u}\big(Q,\dom(V_j),\H_0\big)$ (see
\cite[Sec.~5.2.2]{ABG96}). Thus, the map
$$
\R^d\ni x\mapsto V_j(x))E^{V^2}(I)\in\B(\H_0)
$$
is differentiable in the topology of $\B(\H_0)$, with derivative
$\frac\d{\d x_k}V_j(x)E^{V^2}(I)=V_{jk}'(x)E^{V^2}(I)$. Using this fact, Lemma
\ref{lemma_V}(e) and an integration by parts, one obtains that
\begin{align*}
\frac\d{\d x_k}\int_0^\nu\d s\,\big(x\cdot V(sx)\big)E^{V^2}(I)
&=\int_0^\nu\d s\,V_k(sx)E^{V^2}(I)
+\int_0^\nu\d s\,\sum_{j=1}^dx_js\;\!V_{jk}'(sx))E^{V^2}(I)\\
&=\int_0^\nu\d s\,V_k(sx)E^{V^2}(I)
+\int_0^\nu\d s\,s\sum_{j=1}^dx_jV_{kj}'(sx))E^{V^2}(I)\\
&=\int_0^\nu\d s\,V_k(sx)E^{V^2}(I)
+\int_0^\nu\d s\,s\;\!\tfrac\d{\d s}\;\!V_k(sx)E^{V^2}(I)\\
&=\nu\;\!V_k(\nu x)E^{V^2}(I),
\end{align*}
which proves the claim.
\end{proof}

The next theorem is the main result of this section; it relates the evolution of the
localisation operator $f(Q)$ under $U_0$ to the operator $T_f$.

\begin{Theorem}[Summation formula]\label{thm_summation}
Let Assumptions \ref{ass_regularity} and \ref{ass_commute} be satisfied, and let
$f\in\SS(\R^d)$ be even and equal to $1$ on a neighbourhood of $0\in\R^d$. Then, we
have for each $\varphi\in\dom_2$
\begin{equation}\label{eq_sum}
\lim_{\nu\searrow0}\tfrac12\sum_{n\ge0}\big\langle\varphi,\big(U_0^{-n}f(\nu Q)U_0^n
-U_0^nf(\nu Q)U_0^{-n}\big)\varphi\big\rangle_{\H_0}=t_f(\varphi).
\end{equation}
\end{Theorem}

Note that the sum on the l.h.s. of \eqref{eq_sum} is finite for each $\nu>0$ because
$f(\nu Q)$ can be factorised as
$$
f(\nu Q)
=\big|f(\nu Q)\big|^{1/2}\cdot\hbox{sgn}\big(f(\nu Q)\big)\big|f(\nu Q)\big|^{1/2},
$$
with $\big|f(\nu Q)\big|^{1/2}$ locally $U_0$-smooth on $\S^1\setminus\kappa(U_0)$ due
to Theorem \ref{thm_spectrum}(b). Furthermore, since Remark \ref{remark_T_f}(b)
applies, the r.h.s. of \eqref{eq_sum} can also be written as the expectation value
$\big\langle\varphi,T_f\;\!\varphi\big\rangle_{\H_0}$.

\begin{proof}
(i) Let $\varphi\in\dom_2$ and $\nu>0$. Then, there exist a real function
$\eta\in C^\infty_{\rm c}\big((0,\infty)\big)$ and a compact set $I\subset(0,\infty)$
such that $\varphi=\eta\big(V^2\big)\varphi=E^{V^2}(I)\varphi$, and it follows from
Lemma \ref{lemma_trotter} that
\begin{align}
&\sum_{n\ge0}\big\langle\varphi,\big(U_0^{-n}f(\nu Q)U_0^n
-U_0^nf(\nu Q)U_0^{-n}\big)\varphi\big\rangle_{\H_0}\nonumber\\
&=\sum_{n\ge0}\int_{\R^d}\underline\d x\,(\F f)(x)
\left\langle\varphi,\left(\eta\big(V^2\big)\e^{i\nu x\cdot Q}
\e^{in\int_0^\nu\d s\,(x\cdot V(sx))E^{V^2}(I)}\right.\right.\nonumber\\
&\left.\left.\quad-\e^{-in\int_{-\nu}^0\d s\,(x\cdot V(sx))E^{V^2}(I)}
\e^{i\nu x\cdot Q}\eta\big(V^2\big)\right)\varphi\right\rangle_{\H_0}\nonumber\\
&=\sum_{n\ge0}\int_{\R^d}\underline\d x\,(\F f)(x)
\left\langle\varphi,\left(\e^{i\nu x\cdot Q}\eta\big(V(\nu x)^2\big)
\e^{in\int_0^\nu\d s\,(x\cdot V(sx))E^{V^2}(I)}\right.\right.\nonumber\\
&\left.\left.\quad-\e^{-in\int_{-\nu}^0\d s\,(x\cdot V(sx))E^{V^2}(I)}
\eta\big(V(-\nu x)^2\big)\e^{i\nu x\cdot Q}\right)
\varphi\right\rangle_{\H_0}\nonumber\\
&=\sum_{n\ge0}\int_{\R^d}\underline\d x\,(\F f)(x)
\left\langle\varphi,\left(\big(\e^{i\nu x\cdot Q}-1\big)\eta\big(V(\nu x)^2\big)
\e^{in\int_0^\nu\d s\,(x\cdot V(sx))E^{V^2}(I)}\right.\right.\nonumber\\
&\left.\left.\quad-\e^{-in\int_{-\nu}^0\d s\,(x\cdot V(sx))E^{V^2}(I)}
\eta\big(V(-\nu x)^2\big)\big(\e^{i\nu x\cdot Q}-1\big)\right)
\varphi\right\rangle_{\H_0}\nonumber\\
&\quad+\sum_{n\ge0}\int_{\R^d}\underline\d x\,(\F f)(x)
\left\langle\varphi,\left(\eta\big(V(\nu x)^2\big)
\e^{in\int_0^\nu\d s\,(x\cdot V(sx))E^{V^2}(I)}\right.\right.\nonumber\\
&\left.\left.\quad-\e^{-in\int_{-\nu}^0\d s\,(x\cdot V(sx))E^{V^2}(I)}
\eta\big(V(-\nu x)^2\big)\right)
\varphi\right\rangle_{\H_0}.\label{eq_even}
\end{align}
But, by using the change of variable $x'=-x$ and the fact $\F f$ is even, one obtains
that the second term in \eqref{eq_even} is equal to zero. Thus,

\begin{align}
&\lim_{\nu\searrow0}\sum_{n\ge0}\big\langle\varphi,\big(U_0^{-n}f(\nu Q)U_0^n
-U_0^nf(\nu Q)U_0^{-n}\big)\varphi\big\rangle_{\H_0}\nonumber\\
&=\lim_{\nu\searrow0}\sum_{n\ge0}\int_{\R^d}\underline\d x\,(\F f)(x)
\left\langle\varphi,\left(\big(\e^{i\nu x\cdot Q}-1\big)\eta\big(V(\nu x)^2\big)
\e^{in\int_0^\nu\d s\,(x\cdot V(sx))E^{V^2}(I)}\right.\right.\nonumber\\
&\left.\left.\quad-\e^{-in\int_{-\nu}^0\d s\,(x\cdot V(sx))E^{V^2}(I)}
\eta\big(V(-\nu x)^2\big)\big(\e^{i\nu x\cdot Q}-1\big)\right)
\varphi\right\rangle_{\H_0},\label{eq_Lebesgue_1}
\end{align}
and in point (ii) below we show that we can replace the sum over $n$ by an integral
over $t:$
\begin{align*}
\eqref{eq_Lebesgue_1}
&=\lim_{\nu\searrow0}\int_0^\infty\d t\int_{\R^d}\underline\d x\,(\F f)(x)
\left\langle\varphi,\left(\big(\e^{i\nu x\cdot Q}-1\big)\eta\big(V(\nu x)^2\big)
\e^{it\int_0^\nu\d s\,(x\cdot V(sx))E^{V^2}(I)}\right.\right.\\
&\left.\left.\quad-\e^{-it\int_{-\nu}^0\d s\,(x\cdot V(sx))E^{V^2}(I)}
\eta\big(V(-\nu x)^2\big)\big(\e^{i\nu x\cdot Q}-1\big)\right)
\varphi\right\rangle_{\H_0}.
\end{align*}
Thus, using the change of variable $\mu:=\nu t$, we get
\begin{align*}
&\lim_{\nu\searrow0}\sum_{n\ge0}\big\langle\varphi,\big(U_0^{-n}f(\nu Q)U_0^n
-U_0^nf(\nu Q)U_0^{-n}\big)\varphi\big\rangle_{\H_0}\\
&=\lim_{\nu\searrow0}\int_0^\infty\d\mu\int_{\R^d}\underline\d x\,
(\F f)(x)\left\langle\varphi,\left(\tfrac1\nu\big(\e^{i\nu x\cdot Q}-1\big)
\eta\big(V(\nu x)^2\big)\e^{i\frac\mu\nu\int_0^\nu\d s\,(x\cdot V(sx))E^{V^2}(I)}
\right.\right.\\
&\left.\left.\quad-\e^{-i\frac\mu\nu\int_{-\nu}^0\d s\,(x\cdot V(sx))E^{V^2}(I)}
\eta\big(V(-\nu x)^2\big)\tfrac1\nu\big(\e^{i\nu x\cdot Q}-1\big)\right)
\varphi\right\rangle_{\H_0},
\end{align*}
and in point (iii) below we show that we can exchange the limit $\lim_{\nu\searrow0}$
with the integrals over $\mu$ and $x$ in the last expression. This, together with the
fact that $\F f$ is even, Lemma \ref{lemma_R_f}(a) and Proposition \ref{prop_T_f},
implies that
\begin{align*}
&\lim_{\nu\searrow0}\sum_{n\ge0}\big\langle\varphi,\big(U_0^{-n}f(\nu Q)U_0^n
-U_0^nf(\nu Q)U_0^{-n}\big)\varphi\big\rangle_{\H_0}\\
&=i\int_0^\infty\d\mu\int_{\R^d}\underline\d x\,(\F f)(x)
\left(\left\langle\big(x\cdot Q\big)\varphi,\e^{i\mu(x\cdot V)E^{V^2}(I)}
\varphi\right\rangle_{\H_0}
-\left\langle\varphi,\e^{-i\mu(x\cdot V)E^{V^2}(I)}\big(x\cdot Q\big)
\varphi\right\rangle_{\H_0}\right)\\
&=\sum_{j=1}^d\int_0^\infty\d\mu\int_{\R^d}\underline\d x\,
\big(\F\big(\partial_jf\big)\big)(x)\left(\left\langle Q_j\varphi,
\e^{i\mu(x\cdot V)}\varphi\right\rangle_{\H_0}
+\left\langle\varphi,\e^{i\mu(x\cdot V)}Q_j\varphi\right\rangle_{\H_0}\right)\\
&=\sum_{j=1}^d\int_0^\infty\d\mu\left(\big\langle Q_j\varphi,
\big(\partial_jf\big)(\mu V)\varphi\big\rangle_{\H_0}
+\big\langle\big(\partial_j\overline f\big)(\mu V)\varphi,
Q_j\varphi\big\rangle_{\H_0}\right)\\
&=2\;\!t_f(\varphi).
\end{align*}

(ii) We show here that
\begin{align*}
&\lim_{\nu\searrow0}\sum_{n\ge0}\int_{\R^d}\underline\d x\,(\F f)(x)
\left\langle\varphi,\left(\big(\e^{i\nu x\cdot Q}-1\big)\eta\big(V(\nu x)^2\big)
\e^{in\int_0^\nu\d s\,(x\cdot V(sx))E^{V^2}(I)}\right.\right.\\
&\left.\left.\quad-\e^{-in\int_{-\nu}^0\d s\,(x\cdot V(sx))E^{V^2}(I)}
\eta\big(V(-\nu x)^2\big)\big(\e^{i\nu x\cdot Q}-1\big)\right)
\varphi\right\rangle_{\H_0}\\
&=\lim_{\nu\searrow0}\int_0^\infty\d t\int_{\R^d}\underline\d x\,(\F f)(x)
\left\langle\varphi,\left(\big(\e^{i\nu x\cdot Q}-1\big)\eta\big(V(\nu x)^2\big)
\e^{it\int_0^\nu\d s\,(x\cdot V(sx))E^{V^2}(I)}\right.\right.\\
&\left.\left.\quad-\e^{-it\int_{-\nu}^0\d s\,(x\cdot V(sx))E^{V^2}(I)}
\eta\big(V(-\nu x)^2\big)\big(\e^{i\nu x\cdot Q}-1\big)\right)
\varphi\right\rangle_{\H_0},
\end{align*}
which is equivalent to
\begin{align}\label{eq_exchange_1}
&\lim_{\nu\searrow0}\sum_{n\ge0}\int_0^1\d t\int_{\R^d}\underline\d x\,(\F f)(x)\\
&\quad\cdot\left\langle\varphi,\left(\big(\e^{i\nu x\cdot Q}-1\big)
\left(1-\e^{it\int_0^\nu\d s\,(x\cdot V(sx))E^{V^2}(I)}\right)\eta\big(V(\nu x)^2\big)
\e^{in\int_0^\nu\d s\,(x\cdot V(sx))E^{V^2}(I)}\right.\right.\nonumber\\
&\left.\left.\quad-\e^{-in\int_{-\nu}^0\d s\,(x\cdot V(sx))E^{V^2}(I)}
\eta\big(V(-\nu x)^2\big)
\left(1-\e^{-it\int_{-\nu}^0\d s\,(x\cdot V(sx))E^{V^2}(I)}\right)
\big(\e^{i\nu x\cdot Q}-1\big)\right)\varphi\right\rangle_{\H_0}=0.\nonumber
\end{align}
For this, it is sufficient to prove that we can exchange in \eqref{eq_exchange_1} the
limit $\lim_{\nu\searrow0}$ with the sum over $n$ and the integrals over $t$ and $x$.
We present the calculations only for the first term on the l.h.s. of
\eqref{eq_exchange_1}, since the second term can be handled in a similar way. So, let
\begin{align*}
T(\nu,n)&:=\int_0^1\d t\int_{\R^d}\underline\d x\,(\F f)(x)\\
&\quad\cdot\left\langle\varphi,\big(\e^{i\nu x\cdot Q}-1\big)
\left(1-\e^{it\int_0^\nu\d s\,(x\cdot V(sx))E^{V^2}(I)}\right)\eta\big(V(\nu x)^2\big)
\e^{in\int_0^\nu\d s\,(x\cdot V(sx))E^{V^2}(I)}\varphi\right\rangle_{\H_0}.
\end{align*}
Since $\F f\in\SS(\R^d)$ and
$$
\left\|\big(\e^{i\nu x\cdot Q}-1\big)
\left(1-\e^{it\int_0^\nu\d s\,(x\cdot V(sx))E^{V^2}(I)}\right)\eta\big(V(\nu x)^2\big)
\e^{in\int_0^\nu\d s\,(x\cdot V(sx))E^{V^2}(I)}\right\|_{\B(\H_0)}
\le{\rm Const.},
$$
we have that
\begin{equation}\label{eq_bound_n_0}
|T(\nu,n)|\le{\rm Const.},
\end{equation}
and thus $T(\nu,n)$ is uniformly bounded in $\nu>0$ by a function in
$\ell^1(\{1,\ldots,n_0\})$ for any $n_0\in\N^*$.

For the case $n>n_0$, let
$B_{\nu,n}^I(x):=\e^{in\int_0^\nu\d s\,(x\cdot V(sx))E^{V^2}(I)}$. Then, Lemma
\ref{lemma_trotter}(b) implies that
$$
\big(\partial_jB_{\nu,n}^I\big)(x)\varphi
=in\nu\;\!V_j(\nu x)E^{V^2}(I)B_{\nu,n}^I(x)\varphi
=in\nu\;\!V_j(\nu x)B_{\nu,n}^I(x)\varphi,
$$
and thus $T(\nu,n)$ can be written as
$$
T(\nu,n)
=\frac\nu{in}\sum_{j=1}^d\int_0^1\d t\int_{\R^d}\underline\d x\,
\left\langle\langle Q\rangle^2\varphi,A_{j,\nu,t}(x)
\big(\partial_jB_{\nu,n}^I\big)(x)\varphi\right\rangle_{\H_0},
$$
with
$$
A_{j,\nu,t}(x)
:=(\F f)(x)\;\!\tfrac1\nu\big(\e^{i\nu x\cdot Q}-1\big)\langle Q\rangle^{-2}
\;\!\tfrac1\nu\left(1-\e^{it\int_0^\nu\d s\,(x\cdot V(sx))E^{V^2}(I)}\right)
V_j(\nu x)V(\nu x)^{-2}\eta\big(V(\nu x)^2\big).
$$
Now, for each multi-index $\alpha\in\N^d$ with $|\alpha|\le2$, we have
\begin{equation}\label{eq_bound_Q}
\left\|\partial_x^\alpha\;\!\tfrac1\nu\big(\e^{i\nu x\cdot Q}-1\big)
\langle Q\rangle^{-2}\right\|_{\B(\H_0)}\le{\rm Const.}\;\!\langle x\rangle.
\end{equation}
Thus, it follows from Assumption \ref{ass_regularity}, Lemma \ref{lemma_trotter}(b),
\cite[Prop.~5.1]{RT12_0} and the rapid decay of $\F f$ that the map
$\R^d\ni x\mapsto A_{j,\nu,t}(x)\in\B(\H_0)$ is twice strongly differentiable, with
strong derivatives satisfying for all $j,k\in\{1,\ldots,d\}$ and all $m\in\N$
$$
\big\|\big(\partial_jA_{j,\nu,t}\big)(x)\big\|_{\B(\H_0)}
\le{\rm Const.}\;\!\langle t\rangle\langle x\rangle^{-m}
$$
and
\begin{equation}\label{eq_second_A}
\big\|\partial_k\big\{\big(\partial_jA_{j,\nu,t}\big)V_k(\nu\;\!\cdot\;\!)
V(\nu\;\!\cdot\;\!)^{-2}\big\}(x)\big\|_{\B(\H_0)}
\le{\rm Const.}\;\!\langle\nu\rangle\langle t\rangle^2\langle x\rangle^{-m}.
\end{equation}
So, we can perform two successive integrations by parts with vanishing boundary
contributions to get
\begin{align*}
T(\nu,n)
&=\frac{i\nu}n\sum_{j=1}^d\int_0^1\d t\int_{\R^d}\underline\d x\,
\left\langle\langle Q\rangle^2\varphi,\big(\partial_jA_{j,\nu,t}\big)(x)
B_{\nu,n}^I(x)\varphi\right\rangle_{\H_0}\\
&=-\frac1{n^2}\sum_{j,k=1}^d\int_0^1\d t\int_{\R^d}\underline\d x\,
\left\langle\langle Q\rangle^2\varphi,\partial_k\big\{\big(\partial_jA_{j,\nu,t}\big)
V_k(\nu\;\!\cdot\;\!)V(\nu\;\!\cdot\;\!)^{-2}\big\}(x)B_{\nu,n}^I(x)
\varphi\right\rangle_{\H_0}.
\end{align*}
Combining this with the bound \eqref{eq_second_A}, we get for any $\nu\in(0,1)$ and
$n>n_0$ that
$$
|T(\nu,n)|\le{\rm Const.}\;\!n^{-2}.
$$
This, together with the bound \eqref{eq_bound_n_0}, implies that that $T(\nu,n)$ is
uniformly bounded in $\nu\in(0,1)$ by a function in $\ell^1(\N)$. Thus, we can apply
Lebesgue's dominated convergence theorem to exchange the limit $\lim_{\nu\searrow0}$
with the sum over $n$ in \eqref{eq_exchange_1}. Since the exchange of the limit
$\lim_{\nu\searrow0}$ with the integrals over $t$ and $x$ in \eqref{eq_exchange_1} is
trivial, the result follows.

(iii) We show here that we can exchange the limit $\lim_{\nu\searrow0}$ with the
integrals over $\mu$ and $x$ in the expression
\begin{align}
&\lim_{\nu\searrow0}\int_0^\infty\d\mu\int_{\R^d}\underline\d x\,
(\F f)(x)\left\langle\varphi,\left(\tfrac1\nu\big(\e^{i\nu x\cdot Q}-1\big)
\eta\big(V(\nu x)^2\big)\e^{i\frac\mu\nu\int_0^\nu\d s\,(x\cdot V(sx))E^{V^2}(I)}
\right.\right.\nonumber\\
&\left.\left.\quad-\e^{-i\frac\mu\nu\int_{-\nu}^0\d s\,(x\cdot V(sx))E^{V^2}(I)}
\eta\big(V(-\nu x)^2\big)\tfrac1\nu\big(\e^{i\nu x\cdot Q}-1\big)\right)
\varphi\right\rangle_{\H_0}.\label{eq_exchange_2}
\end{align}
We present the calculations only for the first term in \eqref{eq_exchange_2}, since
the second term can be handled in a similar way. So, let
$$
\widetilde T(\nu,\mu)
:=\int_{\R^d}\underline\d x\,\left\langle\varphi,
(\F f)(x)\;\!\tfrac1\nu\big(\e^{i\nu x\cdot Q}-1\big)\eta\big(V(\nu x)^2\big)
\e^{i\frac\mu\nu\int_0^\nu\d s\,(x\cdot V(sx))E^{V^2}(I)}\varphi\right\rangle_{\H_0}
$$
Due to the bound \eqref{eq_bound_Q} and the inclusion $\F f\in\SS(\R^d)$, we have
\begin{equation}\label{eq_bound_mu}
\big|\widetilde T(\nu,\mu)\big|\le{\rm Const.},
\end{equation}
and thus $\widetilde T(\nu,\mu)$ is uniformly bounded in $\nu>0$ by a function in
$\lone\big((0,1],\d\mu\big)$.

For the case $\mu>1$, let
$
\widetilde B_{\nu,\mu}^I(x)
:=\e^{i\frac\mu\nu\int_0^\nu\d s\,(x\cdot V(sx))E^{V^2}(I)}
$.
Then, Lemma \ref{lemma_trotter}(b) implies that
$$
\big(\partial_j\widetilde B_{\nu,n}^I\big)(x)\varphi
=i\mu\;\!V_j(\nu x)E^{V^2}(I)\widetilde B_{\nu,n}^I(x)\varphi
=i\mu\;\!V_j(\nu x)\widetilde B_{\nu,n}^I(x)\varphi,
$$
and thus $\widetilde T(\nu,\mu)$ can be written as
\begin{align*}
\widetilde T(\nu,\mu)
=\frac1{i\mu}\sum_{j=1}^d\int_{\R^d}\underline\d x\,
\left\langle\langle Q\rangle^2\varphi,\widetilde A_{j,\nu}(x)
\big(\partial_j\widetilde B_{\nu,n}^I\big)(x)\varphi\right\rangle_{\H_0},
\end{align*}
with
$$
\widetilde A_{j,\nu}(x)
:=(\F f)(x)\;\!\tfrac1\nu\big(\e^{i\nu x\cdot Q}-1\big)\langle Q\rangle^{-2}
\;\!V_j(\nu x)V(\nu x)^{-2}\eta\big(V(\nu x)^2\big).
$$
Furthermore, one can show as in point (ii) that that the map
$\R^d\ni x\mapsto \widetilde A_{j,\nu}(x)\in\B(\H_0)$ is twice strongly
differentiable, with strong derivatives satisfying for all $j,k\in\{1,\ldots,d\}$ and
all $m\in\N$
$$
\big\|\big(\partial_j\widetilde A_{j,\nu}\big)(x)\big\|_{\B(\H_0)}
\le{\rm Const.}\;\!\langle x\rangle^{-m}
$$
and
\begin{equation}\label{eq_second_tilde_A}
\big\|\partial_k\big\{\big(\partial_j\widetilde A_{j,\nu}\big)V_k(\nu\;\!\cdot\;\!)
V(\nu\;\!\cdot\;\!)^{-2}\big\}(x)\big\|_{\B(\H_0)}
\le{\rm Const.}\;\!\langle\nu\rangle\langle x\rangle^{-m}.
\end{equation}
So, we can perform two successive integrations by parts with vanishing boundary
contributions to get
\begin{align*}
\widetilde T(\nu,\mu)
&=\frac i\mu\sum_{j=1}^d\int_{\R^d}\underline\d x\,
\left\langle\langle Q\rangle^2\varphi,\big(\partial_j\widetilde A_{j,\nu,t}\big)(x)
\widetilde B_{\nu,n}^I(x)\varphi\right\rangle_{\H_0}\\
&=-\frac1{\mu^2}\sum_{j,k=1}^d\int_{\R^d}\underline\d x\,
\left\langle\langle Q\rangle^2\varphi,\partial_k
\big\{\big(\partial_j\widetilde A_{j,\nu,t}\big)V_k(\nu\;\!\cdot\;\!)
V(\nu\;\!\cdot\;\!)^{-2}\big\}(x)\widetilde B_{\nu,n}^I(x)\varphi\right\rangle_{\H_0}.
\end{align*}
Combining this with the bound \eqref{eq_second_tilde_A}, we get for any $\nu\in(0,1)$
and $\mu>1$ that
$$
\big|\widetilde T(\nu,\mu)\big|\le{\rm Const.}\;\!\mu^{-2}.
$$
This, together with the bound \eqref{eq_bound_mu}, implies that that
$\widetilde T(\nu,\mu)$ is uniformly bounded in $\nu\in(0,1)$ by a function in
$\lone\big((0,\infty),\d\mu\big)$. Thus, we can apply Lebesgue's dominated convergence
theorem to exchange the limit $\lim_{\nu\searrow0}$ with the integral over $\mu$ in 
\eqref{eq_exchange_2}. Since the exchange of the limit $\lim_{\nu\searrow0}$ with the
integral over $x$ in \eqref{eq_exchange_2} is trivial, the result follows.
\end{proof}

\subsection{Interpretation of the summation formula}\label{section_interpretation}

In this section, we explain why the operator $T_f$ can be considered as a time
operator for $U_0$ and we give an interpretation of the summation formula
\eqref{eq_sum}. We start with a lemma which establishes crucial commutation relations
between the operators $T_f$ and $U_0:$

\begin{Lemma}\label{lemma_canonical}
Let Assumptions \ref{ass_regularity} and \ref{ass_commute} be satisfied, and let
$f\in\SS(\R^d)$ be real and equal to $1$ on a neighbourhood of $0\in\R^d$.
\begin{enumerate}
\item[(a)] We have
\begin{equation}\label{eq_weyl_1}
T_fU_0^n\varphi=\big(U_0^nT_f-n\;\!U_0^n\big)\varphi,\quad n\in\Z,~\varphi\in\dom_1.
\end{equation}
\item[(b)] If $\dom_1$ is dense in $\H_0$ and $T_f$ is essentially self-adjoint on
$\dom_1$ with closure $\overline{T_f}$, then we have the imprimitivity relation
\begin{equation}\label{eq_weyl_2}
\e^{is\overline{T_f}}\gamma(U_0)\e^{-is\overline{T_f}}=\gamma(\e^{-is}U_0),
\quad s\in\R,~\gamma\in C(\S^1).
\end{equation}
\end{enumerate}
\end{Lemma}

\begin{proof}
(a) Since $\dom(\langle Q\rangle)=\cap_{j=1}^d\dom(Q_j)$, we have the equality
$$
\dom_1=\left\{\varphi\in\cap_{j=1}^d\dom(Q_j)
\mid\varphi=\zeta(U_0)\varphi=\rho(D)\varphi
\hbox{ for some $\zeta\in C^\infty_{\rm c}\big(\S^1\setminus\kappa(U_0)\big)$ and
$\rho\in C^\infty_{\rm c}(\R)$}\right\},
$$
and thus $U_0^n\dom_1\subset\dom_1$ for each $n\in\Z$ due to Lemma \ref{lemma_V}(c)
and the definition of the operator $D$. Moreover, if $\varphi\in\dom_1$ and
$f\in\SS(\R^d)$, then we have the inclusions
$\big(\partial_jR_f\big)(V)\varphi\in\dom(Q_j)\cap\dom(V_j)$ and
$|V|^{-1}\varphi\in\dom(Q_j)\cap\dom(V_j)$ for each $j\in\{1,\ldots,d\}$. Therefore,
using successively the strong commutation of $U_0$ and $V$, Lemma \ref{lemma_V}(c),
and the relations \eqref{eq_homo_1}-\eqref{eq_homo_2}, we obtain
\begin{align*}
&\left(Q\cdot(\nabla R_f)(V)
+(\nabla R_f)\big(\tfrac V{|V|}\big)\cdot Q\;\!|V|^{-1}\right)U_0^n\varphi\\
&=\sum_{j=1}^d\left(Q_jU_0^n\big(\partial_jR_f\big)(V)
+\big(\partial_jR_f\big)\big(\tfrac V{|V|}\big)Q_jU_0^n|V|^{-1}\right)\varphi\\
&=U_0^n\sum_{j=1}^d\left((Q_j+n\;\!V_j)\big(\partial_jR_f\big)(V)
+\big(\partial_jR_f\big)\big(\tfrac V{|V|}\big)(Q_j+n\;\!V_j)|V|^{-1}\right)\varphi\\
&=U_0^n\left(Q\cdot(\nabla R_f)(V)
+(\nabla R_f)\big(\tfrac V{|V|}\big)\cdot Q\;\!|V|^{-1}\right)\varphi
-2n\;\!U_0^n\varphi.
\end{align*}
This, together with \eqref{eq_T_f}, implies that
$T_fU_0^n\varphi=\big(U_0^nT_f-n\;\!U_0^n\big)\varphi$.

(b) We know from \eqref{eq_weyl_1} that $U_0^{-1}T_fU_0\varphi=(T_f-1)\varphi$ for
each $\varphi\in\dom_1$. Since $T_f$ is essentially self-adjoint on $\dom_1$, it
follows that
$$
U_0^{-1}\overline{T_f}U_0=\overline{U_0^{-1}T_fU_0}=\overline{T_f}-1.
$$
Using this relation and functional calculus, we infer that
\begin{align*}
\e^{is\overline{T_f}}\gamma(U_0)\e^{-is\overline{T_f}}
=\gamma\big(U_0\e^{isU_0^{-1}\overline{T_f}U_0}\e^{-is\overline{T_f}}\big)
=\gamma\big(U_0\e^{is(\overline{T_f}-1)}\e^{-is\overline{T_f}}\big)
=\gamma(\e^{-is}U_0),
\end{align*}
which proves the claim.
\end{proof}

If $\dom_1$ is dense in $\H_0$ and $T_f$ is essentially self-adjoint on $\dom_1$, then
\eqref{eq_weyl_2} and Mackey's imprimitivity theorem \cite[Thm.~5]{Ors79} applied to
the group $\R$ and the subgroup $\Z$ imply the existence of a continuous unitary
representation $\sigma$ of $\Z$ in a Hilbert space $\frak h_\sigma$ achieving the
following: Let $F_\sigma$ be the set of functions $f_\sigma:\R\to\frak h_\sigma$ such
that
\begin{enumerate}
\item[(i)] $f_\sigma(n+s)=\sigma(n)f_\sigma(s)$ for all $n\in\Z$ and $s\in\R$,
\item[(ii)] $\|f_\sigma(\;\!\cdot\;\!)\|_{\frak h_\sigma}\in\ltwoloc(\R)$,
\item[(iii)] $f_\sigma$ is strongly measurable,
\end{enumerate}
let $\langle\cdot,\cdot\rangle_{\H_\sigma}$ and $\|\cdot\|_{\H_\sigma}$ be the scalar
product and norm on $F_\sigma$ given by
$$
\langle f_\sigma,g_\sigma\rangle_{\H_\sigma}
:=\int_0^1\d s\,\langle f_\sigma(s),g_\sigma(s)\rangle_{\frak h_\sigma}
\quad\hbox{and}\quad
\|f_\sigma\|_{\H_\sigma}:=\sqrt{\langle f_\sigma,f_\sigma\rangle_{\H_\sigma}},
\quad f_\sigma,g_\sigma\in F_\sigma,
$$
and let $\H_\sigma$ be the Hilbert space completion of $F_\sigma$ for the norm
$\|\cdot\|_{\H_\sigma}$, that is,
$$
\H_\sigma:=\big\{f_\sigma\in F_\sigma\mid\|f_\sigma\|_{\H_\sigma}<\infty\big\}
/\big\{f_\sigma\in F_\sigma\mid\|f_\sigma\|_{\H_\sigma}=0\big\}.
$$
Then, there exists a unitary operator $\U:\H_0\to\H_\sigma$ satisfying for all
$s\in\R$ and $\gamma\in C(\S^1)$
$$
\U\e^{-2\pi is\overline{T_f}}\U^{-1}=U_\sigma(s)
\quad\hbox{and}\quad
\U\gamma(U_0)\U^{-1}=P_\sigma(\gamma),
$$
with $U_\sigma$ the induced continuous unitary representation of $\sigma$ from $\Z$ to
$\R$ given by
$$
\big(U_\sigma(s)f_\sigma\big)(t):=f_\sigma(t+s),\quad s,t\in\R,~f_\sigma\in\H_\sigma,
$$
and $P_\sigma$ given by
$$
\big(P_\sigma(\gamma)f_\sigma\big)(s):=\gamma(\e^{2\pi is})f_\sigma(s),
\quad s\in\R,~f_\sigma\in\H_\sigma,~\gamma\in C(\S^1).
$$
Therefore, the spectrum of $U_0$ is purely absolutely continuous and covers the whole
unit circle $\S^1$, and we get for all $\psi\in\H_0$ and $\varphi\in\dom_1$ the
equalities
$$
\big\langle\psi,T_f\;\!\varphi\big\rangle_{\H_0}
=\big\langle\psi,\overline{T_f}\;\!\varphi\big\rangle_{\H_0}
=\int_0^1\d s\,\left\langle(\U\psi)(s),
\frac i{2\pi}\frac{\d(\U\varphi)}{\d s}(s)\right\rangle_{\frak h_\sigma},
$$
with $\frac{\d(\U\varphi)}{\d s}(s)$ the distributional derivative at $s$ of the
function $\R\ni s\mapsto(\U\varphi)(s)\in\frak h_\sigma$. In particular, if we make
the change of variable $z(s):=\e^{2\pi i s}\in\S^1$ and choose functions
$\widetilde{\U\psi},\widetilde{\U\varphi}:\S^1\to\frak h_\sigma$ satisfying
$$
(\U\psi)(s)=\big(\widetilde{\U\psi}\big)(z(s))
\quad\hbox{and}\quad
(\U\varphi)(s)=\big(\widetilde{\U\varphi}\big)(z(s))
$$
for each $s\in[0,1)$, we obtain the identity
\begin{equation}\label{eq_derivative}
\big\langle\psi,T_f\;\!\varphi\big\rangle_{\H_0}
=\int_{\S^1}\d\mu_{\S^1}(z)\left\langle\big(\widetilde{\U\psi}\big)(z),
-z\;\!\frac{\d\big(\widetilde{\U\varphi}\big)}{\d z}(z)\right\rangle_{\frak h_\sigma}
\end{equation}
with $\d\mu_{\S^1}(z):=\frac{\d z}{2\pi iz}$ the Haar measure on $\S^1$.

If $\dom_1$ is dense in $\H_0$, then Proposition \ref{prop_T_f} and Remark
\ref{remark_T_f}(b) imply that $T_f$ is symmetric, and the relations
$\H_0=\overline{\dom_1}\subset\H_{\rm ac}(U_0)$ imply that $U_0$ has purely absolutely
continuous spectrum. However, the spectrum of $U_0$ may not cover the whole unit
circle $\S^1$. Either way, we expect that the operator $T_f$ is still equal to a
differential operator in some Hilbert space isomorphic to $\H_0$, but we have not been
able to prove it in this generality.

If $\dom_1$ is not dense in $\H_0$, then we are not aware of works using a relation
like \eqref{eq_weyl_1} to infer results on the spectral nature of $U_0$ or on the form
of $T_f$. In such a case, we only know from Theorem \ref{thm_spectrum}(a) that $U_0$
has purely absolutely continuous spectrum in $\sigma(U_0)\setminus\kappa(U_0)$.
However, if one makes some additional assumption on the action of $T_f$ on $\dom_1$,
one should be able to obtain further results on $U_0$ and $T_f$. We refrain to do it
here, but we refer to \cite[p.~324]{RT12_0} for a discussion of this issue in the
self-adjoint setup.

\begin{Remark}[Interpretation of the summation formula]\label{rem_interpretation}
The results that precede have a nice physical interpretation. Lemma
\ref{lemma_canonical}(a) implies that the operators $T_f$ and $U_0$ satisfy on
$\dom_1$ the relation
$$
U_0^{-1}\big[T_f,U_0\big]=-1,
$$
which is the unitary analogue of the canonical time-energy commutation relation of the
self-adjoint setup. Accordingly, the operator $T_f$ can be interpreted as a time
operator for $U_0$, and $T_f$ should be equal in some suitable sense to the operator
$-U_0\;\!\frac\d{\d U_0}$. Indeed, this is essentially what tells us Equation
\eqref{eq_derivative}: if $\dom_1$ is dense in $\H_0$ and $T_f$ is essentially
self-adjoint on $\dom_1$, then $T_f$ acts (after a change of variable) as the
differential operator $-z\;\!\frac\d{\d z}$ ($z\in\S^1$) in the Hilbert space
$\H_\sigma$ isomorphic to $\H_0$.

On another hand, the l.h.s. of Formula \eqref{eq_sum} has the following meaning: For
$\nu>0$ fixed, it can be interpreted as the difference of times spent by the evolving
state $U_0^n\varphi$ in the future (first term) and in the past (second term) within
the region defined by the localisation operator $f(\nu Q)$. Therefore, Formula
\eqref{eq_sum} shows that this difference of times tends as $\nu\searrow0$ to the
expectation value in $\varphi$ of the time operator $T_f$.
\end{Remark}

We conclude this section with an illustration of these results in the setups of
Examples \ref{ex_V_constant} and \ref{ex_Laplacian}.

\begin{Example}[$V$ constant, continued]\label{V_constant_continued}
If $Q$ and $U_0$ are such that $V=v\in\R^d\setminus\{0\}$ as in Example
\ref{ex_V_constant}, then we have $\kappa(U_0)=\varnothing$. Therefore, if we set
$D=V^2=v^2$, we obtain that
\begin{align*}
\dom_1
&=\left\{\varphi\in\dom(\langle Q\rangle)
\mid\varphi=\zeta(U_0)\varphi=\rho(v^2)\varphi
\hbox{ for some $\zeta\in C^\infty_{\rm c}(\S^1)$ and
$\rho\in C^\infty_{\rm c}(\R)$}\right\}\\
&=\dom(\langle Q\rangle)\\
&=\cap_{j=1}^d\dom(Q_j).
\end{align*}
Since \eqref{eq_T_f} implies the equality $T_f=(\nabla R_f)(v)\cdot Q$ on $\dom_1$, it
follows that $T_f$ is essentially self-adjoint on $\dom_1$, and thus
\eqref{eq_derivative} implies that $T_f$ acts (after a change of variable) as the
differential operator $-z\;\!\frac\d{\d z}$ ($z\in\S^1$) in the Hilbert space
$\H_\sigma$.
\end{Example}

\begin{Example}[Time-one propagator for the Laplacian, continued]
\label{Laplacian_continued}
If $Q$, $P$, $U_0$ and $\H_0$ are as in Example \ref{ex_Laplacian}, then we have
$V=2P$ and $\kappa(U_0)=\{1\}$. Furthermore, if we take $f$ radial and set
$D=V^2=4P^2$, then the set
$$
\dom_1=\left\{\varphi\in\dom(\langle Q\rangle)
\mid\varphi=\zeta(\e^{-iP^2})\varphi=\rho(P^2)\varphi
\hbox{ for some $\zeta\in C^\infty_{\rm c}(\S^1\setminus\{1\})$ and
$\rho\in C^\infty_{\rm c}(\R)$}\right\}
$$
is dense in $\H_0$, and a calculation using \eqref{eq_T} shows the following
equalities on $\dom_1$
$$
T_f=T
=-\tfrac14\left(Q\cdot\tfrac P{P^2}+\tfrac P{|P|}\cdot Q\;\!|P|^{-1}+iP^{-2}\right)
=-\tfrac14\left(Q\cdot\tfrac P{P^2}+\tfrac P{P^2}\cdot Q\right).
$$
Thus, it follows from \cite[p.~484-485]{AC87} that $T$ is symmetric on $\dom_1$ and
acts as the differential operator $-z\;\!\frac\d{\d z}$ ($z\in\S^1$) in the spectral
representation of $U_0$.
\end{Example}

\section{Quantum time delay}\label{section_delay}
\setcounter{equation}{0}

\subsection{Symmetrised time delay}\label{section_symmetrised}

In this section, we prove the existence of symmetrised time delay for a quantum
scattering system $(U_0,U,J)$ with free propagator $U_0$, full propagator $U$, and
identification operator $J$. The propagator $U_0$ is a unitary operator that acts in
the Hilbert space $\H_0$ and satisfies Assumptions \ref{ass_regularity} and
\ref{ass_commute} with respect to the family of position operators $Q$. The propagator
$U$ is a unitary operator in a Hilbert space $\H$ that satisfies Assumption
\ref{ass_wave} below. The operator $J:\H_0\to\H$ is a bounded operator used to
identify the Hilbert space $\H_0$ with a subset of the Hilbert space $\H$. The
assumption on $U$ asserts the existence, the isometry and the completeness of the
generalised wave operators for the scattering system $(U_0,U,J)$. To state it, we use
the notation $P_{\rm ac}(U_0)$ for the projection onto the subspace $\H_{\rm ac}(U_0)$
of absolute continuity of $U_0$ (and idem for $U$):

\begin{Assumption}[Wave operators]\label{ass_wave}
The wave operators
$$
W_\pm:=\slim_{n\to\pm\infty}U^{-n}J\;\!U_0^nP_{\rm ac}(U_0)
$$
exist and are partial isometries with initial subspaces $\H_0^\pm\subset\H_0$ and
final subspaces $\H_{\rm ac}(U)$.
\end{Assumption}

Sufficient conditions on the difference $JU_0-UJ$ guaranteeing the existence and the
completeness of $W_\pm$ are given, for instance, in \cite[Sec.~2]{RST_2}. The main
consequence of Assumption \ref{ass_wave} is that the scattering operator
$$
S:=W_+^*W_-:\H_0^-\to\H_0^+
$$
is a well-defined unitary operator commuting with $U_0$.

We now define the sojourn times for the scattering system $(U_0,U,J)$, starting with
the sojourn time for the free evolution $\{U_0^n\}_{n\in\Z}$. Given a positive number
$r>0$, a non-negative function $f\in\SS(\R^d)$ equal to $1$ on a neighbourhood
$\Sigma$ of $0\in\R^d$ and a vector $\varphi\in\dom_0$, we set
$$
T_r^0(\varphi)
:=\sum_{n\in\Z}\big\langle U_0^n\varphi,f(Q/r)U_0^n\varphi\big\rangle_{\H_0}.
$$
The operator $f(Q/r)$ is approximately equal to the projection onto the subspace
$E^Q(r\Sigma)\H_0$ of $\H_0$, with $r\Sigma:=\{x\in\R^d\mid x/r\in\Sigma\}$.
Therefore, if $\|\varphi\|_{\H_0}=1$, then $T_r^0(\varphi)$ can be roughly interpreted
as the time spent by the evolving state $U_0^n\varphi$ inside $E^Q(r\Sigma)\H_0$. The
quantity $T_r^0(\varphi)$ is finite for each $\varphi\in\dom_0$, since we know from
Lemma \ref{thm_spectrum}(b) that the operator $|f(Q/r)|^{1/2}$ is locally $U_0$-smooth
on $\S^1\setminus\kappa(U_0)$.

When defining the sojourn time for the full evolution $\{U^n\}_{n\in\Z}$, one faces
the problem that the localisation operator $f(Q/r)$ acts in $\H_0$, while the operator
$U^n$ acts in $\H$. The obvious modification would be to use the operator
$Jf(Q/r)J^*\in\B(\H)$, but the resulting definitions could be not general enough (see
\cite[Rem.~4.5]{RT12} for a discussion of this issue in the case of scattering for
self-adjoint operators). Sticking to the basic idea that the freely evolving state
$U_0^n\varphi$ should approximate, as $n\to\pm\infty$, the corresponding evolving
state $U^nW_\pm\varphi$, one should look for operators $L_n:\H\to\H_0$ satisfying the
condition
$$
\lim_{n\to\pm\infty}\big\|L_nU^nW_\pm\varphi-U_0^n\varphi\big\|_{\H_0}=0.
$$
Since we consider vectors $\varphi\in\dom_0$, the operators $L_n$ can be unbounded as
long as $L_nU^nW_\pm E^D(I)$ are bounded for all compact sets $I\subset\R$ (if $U_0$
is the time-one propagator of some Hamiltonian $H_0$ and $D=H_0$, then one can simply
require that $L_nE^H(I)$ are bounded for each compact set $I\subset\R$). With these
operators $L_n$ at hand, it is natural to define the sojourn time for the full
evolution $\{U^n\}_{n\in\Z}$ as
$$
T_{r,1}(\varphi):=\sum_{n\in\Z}\big\langle L_nU^nW_-\varphi,
f(Q/r)L_nU^nW_-\varphi\big\rangle_{\H_0}.
$$
Another sojourn time appearing naturally in this context is
$$
T_2(\varphi):=\sum_{n\in\Z}\left(\langle\varphi,\varphi\rangle_{\H_0}
-\big\langle L_nU^nW_-\varphi,L_nU^nW_-\varphi\big\rangle_{\H_0}\right).
$$
The finiteness of $T_{r,1}(\varphi)$ and $T_2(\varphi)$ is proved under some
additional assumptions in Lemma \ref{lemma_free} below.

The term $T_{r,1}(\varphi)$ can be roughly interpreted as the time spent by the
scattering state $U^nW_-\varphi$ inside $E^Q(r\Sigma)\H_0$ after being injected in
$\H_0$ by $L_n$. If some slight abuse of notation is allowed to write the term
$T_2(\varphi)$ as
$$
T_2(\varphi)
=\sum_{n\in\Z}\big\langle U^nW_-\varphi,\big(1-L_n^*L_n\big)U^nW_-\varphi\big\rangle_{\H_0},
$$
then $T_2(\varphi)$ can be interpreted as the time spent by the scattering state
$U^nW_-\varphi$ inside the time-dependent subset $(1-L_n^*L_n)\H$ of $\H$. If $L_n$ is
considered as a time-dependent quasi-inverse for the operator $J$ (see
\cite[Sec.~2.3.2]{Yaf92} for a related notion of time-independent quasi-inverse), then
the subset $(1-L_n^*L_n)\H$ can be interpreted as an approximate complement of $J\H_0$
in $\H$ at time $n$. The necessity of the term $T_2(\varphi)$ in the setup of
two-Hilbert spaces quantum scattering can easily be illustrated when, for example,
$U_0$ and $U$ are time-one propagators of Hamiltonians presenting some multichannel
structure (see for instance \cite[Sec.~5]{RT13_1}). On the other hand, when $\H_0=\H$,
it is natural to set $L_n=J^*=1$, and then $T_2(\varphi)$ vanishes.

Within this general framework, we say that
$$
\tau_r^{\rm sym}(\varphi)
:=T_r(\varphi)-\tfrac12\big(T_r^0(\varphi)+T_r^0(S\varphi)\big),
$$
with $T_r(\varphi):=T_{r,1}(\varphi)+T_2(\varphi)$, is the symmetrised time delay of
the scattering system $(U_0,U,J)$ with incoming state $\varphi$ in the region defined
by the localisation operator $f(Q/r)$, and we say that
$$
\tau_r^{\rm nsym}(\varphi):=T_r(\varphi)-T_r^0(\varphi)
$$
is the non-symmetrised time delay of the scattering system $(U_0,U,J)$ with incoming
state $\varphi$ in the region defined by the localisation operator $f(Q/r)$. In the
case of scattering for self-adjoint operators, the symmetrised time delay is the only
time delay having a well-defined limit as $r\to\infty$ for complicated scattering
systems (see for example \cite{AJ07,BO79,GT07,Mart75,Mar81,RT13_1,SM92,Smi60,Tie06}).

Finally, for the next lemma, we need the auxiliary quantity
\begin{equation}\label{eq_tau_free}
\tau_r^{\rm free}(\varphi):=\tfrac12\sum_{n\ge0}\big\langle\varphi,
S^*\big[U_0^{-n}f(Q/r)U_0^n-U_0^nf(Q/r)U_0^{-n},S\big]\varphi\big\rangle_{\H_0},
\end{equation}
which is finite for all vectors $\varphi\in\H_0^-\cap\dom_0$ satisfying
$S\varphi\in\dom_0$ (see \cite[Eq.~(4.5)]{RT12} for a similar definition in the case
of scattering for self-adjoint operators).

\begin{Lemma}\label{lemma_free}
Let Assumptions \ref{ass_regularity}, \ref{ass_commute} and \ref{ass_wave} be
satisfied. Let $f\in\SS(\R^d)$ be non-negative and equal to $1$ on a neighbourhood of
$0\in\R^d$. For each $n\in\Z$, let $L_n:\H\to\H_0$ satisfy
$L_nU^nW_\pm E^D(I)\in\B(\H,\H_0)$ for all compact sets $I\subset\R$. Finally, let
$\varphi\in\H_0^-\cap\dom_0$ satisfy $S\varphi\in\dom_0$ and
\begin{equation}\label{eq_l_one}
n\mapsto\big\|\big(L_nW_--1\big)U_0^n\varphi\big\|_{\H_0}\in\ell^1(-\N)
\quad\hbox{and}\quad
n\mapsto\big\|\big(L_nW_+-1\big)U_0^nS\varphi\big\|_{\H_0}\in\ell^1(\N).
\end{equation}
Then, $T_r(\varphi)$ is finite for each $r>0$, and
$$
\lim_{r\to\infty}\big(\tau_r^{\rm sym}(\varphi)-\tau_r^{\rm free}(\varphi)\big)=0.
$$
\end{Lemma}

\begin{proof}
The proof consists in showing that the expression
\begin{align}
I_r(\varphi)
&:=T_{r,1}(\varphi)-\tfrac12\big(T_r^0(\varphi)+T_r^0(S\varphi)\big)
-\tau_r^{\rm free}(\varphi)\nonumber\\
&=\sum_{n\le-1}\left(\big\langle L_nU^nW_-\varphi,
f(Q/r)L_nU^nW_-\varphi\big\rangle_{\H_0}
-\big\langle U_0^n\varphi,f(Q/r)U_0^n\varphi\big\rangle_{\H_0}\right)\nonumber\\
&\quad+\sum_{n\ge0}\left(\big\langle L_nU^nW_-\varphi,
f(Q/r)L_nU^nW_-\varphi\big\rangle_{\H_0}
-\big\langle U_0^nS\varphi,f(Q/r)U_0^nS\varphi\big\rangle_{\H_0}\right)\nonumber\\
&\quad-\tfrac12\big\langle\varphi,f(Q/r)\varphi\big\rangle_{\H_0}
+\tfrac12\big\langle S\varphi,f(Q/r)S\varphi\big\rangle_{\H_0}\label{eq_boundary_terms}
\end{align}
converges as $r\to\infty$ to $-T_2(\varphi)$. But, apart from the boundary terms
\eqref{eq_boundary_terms} which cancel out as $r\to\infty$, this can be done as in the
self-adjoint case \cite[Lemma~4.2]{RT12}. So, we leave the details to the reader.
\end{proof}

The next Theorem establishes the existence of the symmetrized time delay; it is a
direct consequence of Theorem \ref{thm_summation}, Definition \eqref{eq_tau_free} and
Lemma \ref{lemma_free}.

\begin{Theorem}[Symmetrised time delay]\label{thm_sym}
Let Assumptions \ref{ass_regularity}, \ref{ass_commute} and \ref{ass_wave} be
satisfied. Let $f\in\SS(\R^d)$ be non-negative, even and equal to $1$ on a
neighbourhood of $0\in\R^d$. For each $n\in\Z$, let $L_n:\H\to\H_0$ satisfy
$L_nU^nW_\pm E^D(I)\in\B(\H,\H_0)$ for all compact sets $I\subset\R$. Finally, let
$\varphi\in\H_0^-\cap\dom_2$ satisfy $S\varphi\in\dom_2$ and \eqref{eq_l_one}. Then,
one has
\begin{equation}\label{eq_EW_sym}
\lim_{r\to\infty}\tau_r^{\rm sym}(\varphi)
=\big\langle\varphi,S^*[T_f,S]\varphi\big\rangle_{\H_0},
\end{equation}
with $T_f$ defined in \eqref{eq_T_f}.
\end{Theorem}

\begin{Remark}\label{rem_sym}
Theorem \ref{thm_sym} is the main result of the paper. It shows the identity of the
symmetrised time delay defined in terms of sojourn times and an analogue of
Eisenbud-Wigner time delay for general unitary scattering systems $(U_0,U,J)$. The
l.h.s. of \eqref{eq_EW_sym} is equal to the global symmetrised time delay of the
scattering system $(U_0,U,J)$, with incoming state $\varphi$, in the dilated regions
defined by the localisation operators $f(Q/r)$. The r.h.s. of \eqref{eq_EW_sym} is the
expectation value in $\varphi$ of the Eisenbud-Wigner-type time delay operator
$S^*[T_f,S]$. When $T_f$ acts in some suitable sense as the differential operator
$-U_0\frac\d{\d U_0}$, which occurs in most of the situations of interest (see Section
\ref{section_interpretation}), one obtains an analogue of Eisenbud-Wigner formula for
unitary scattering systems:
$$
\lim_{r\to\infty}\tau_r^{\rm sym}(\varphi)
=\left\langle\varphi,-S^*U_0\;\!\frac{\d S}{\d U_0}\;\!\varphi\right\rangle_{\H_0}.
$$
\end{Remark}

\subsection{Non-symmetrised time delay}\label{section_non}

We present in this section conditions under which the symmetrised time delay
$\tau_r^{\rm sym}(\varphi)$ and the non-symmetrised time delay
$\tau_r^{\rm nsym}(\varphi)$ are equal in the limit $r\to\infty$. Physically, this
cannot hold if the scattering is not elastic or is of multichannel type. But for
simple scattering systems, the freely evolving states $U_0^n\varphi$ and
$U_0^nS\varphi$ should spend the same time in the region defined by the localisation
operator $f(Q/r)$ in the limit $r\to\infty$, and thus the equality of both time delays
should be verified. Mathematically, this equality reduces to finding conditions under
which
\begin{equation}\label{eq_equal_sojourn}
\lim_{r\to\infty}\left(T_r^0(S\varphi)-T_r^0(\varphi)\right)=0.
\end{equation}
Formally, the proof of \eqref{eq_equal_sojourn} goes as follows: Suppose that the
scattering operator $S$ strongly commutes for each $\nu>0$ with the operator
$\sum_{n\in\Z}f(\nu n\;\!V)$ (i.e. the scattering system is simple in the sense that
it preserves some appropriate function of the velocity vector $V$). Then, using the
change of variables $\nu:=1/r$, one gets
\begin{align*}
&\lim_{r\to\infty}\left(T_r^0(S\varphi)-T_r^0(\varphi)\right)\\
&=\lim_{\nu\searrow0}\sum_{n\in\Z}\left\langle\varphi,
S^*\big[U_0^{-n}f(\nu Q)U_0^n,S\big]\varphi\right\rangle_{\H_0}
-\left\langle\varphi,S^*\big[F_{\nu,f}(V),S\big]\varphi\right\rangle_{\H_0}\\
&=\lim_{\nu\searrow0}\sum_{n\in\Z}\left\langle\varphi,
S^*\big[f(\nu Q+\nu n\;\!V)-f(\nu n\;\!V),S\big]\varphi\right\rangle_{\H_0}\\
&=0.
\end{align*}
A rigorous justification of this argument is given in Proposition
\ref{prop_equal_sojourn} below. Before this, we need two technical lemmas and an
assumption on the behaviour of the $C_0$-group $\{\e^{ix\cdot Q}\}_{x\in\R^d}$ in the
subspace $\dom(\langle V\rangle)$. We start with the first technical lemma:

\begin{Lemma}\label{lemma_poisson}
Let Assumptions \ref{ass_regularity} and \ref{ass_commute} be satisfied. Take
$f\in\SS(\R^d)$, $\eta\in C^\infty_{\rm c}\big((0,\infty)\big)$, and $\varphi\in\H_0$
such that $\varphi=E^{V^2}(I)\varphi$ for some compact set $I\subset(0,\infty)$.
Finally, define for $\nu\in(0,1)$ and $t\in\R$
$$
g_\nu(t):=\int_{\R^d}\underline\d x\,(\F f)(x)\left\langle\varphi,
\left(\eta\big(V(\nu x)^2\big)\e^{it\int_0^\nu\d s\,(x\cdot V(sx))E^{V^2}(I)}
-\eta\big(V^2\big)\e^{i\nu tx\cdot VE^{V^2}(I)}\right)\varphi\right\rangle_{\H_0}.
$$
Then, we have the equality
\begin{equation}\label{eq_poisson}
\lim_{\nu\searrow0}\sum_{n\in\Z}g_\nu(n)
=\lim_{\nu\searrow0}\int_\R\d t\,g_\nu(t).
\end{equation}
\end{Lemma}

\begin{proof}
The proof consists in two steps: In the first step, we show that the function
$g_\nu:\R\to\C$ satisfies the hypotheses of the Poisson summation formula
\cite[Thm.~8.32]{Fol99}, and in the second step we show the equality
\eqref{eq_poisson}.

(i) A direct calculation using the fact that $f\in\SS(\R^d)$ shows that
$g_\nu\in C^\infty$, with $k$-th derivative given by
\begin{align}
g_\nu^{(k)}(t)
&=(i\nu)^k\int_{\R^d}\underline\d x\,(\F f)(x)\left\langle\varphi,
\left(\eta\big(V(\nu x)^2\big)
\left(\int_0^1\d s\,\big(x\cdot V(sx)\big)E^{V^2}(I)\right)^k
\e^{it\int_0^\nu\d s\,(x\cdot V(sx))E^{V^2}(I)}\right.\right.\nonumber\\
&\quad\left.\left.-\eta\big(V^2\big)\big(x\cdot VE^{V^2}(I)\big)^k
\e^{i\nu tx\cdot VE^{V^2}(I)}\right)\varphi\right\rangle_{\H_0}.
\label{eq_derivative_g}
\end{align}
So, in particular $g_\nu$ is continuous.

We now show that there exists $\varepsilon>0$ such that
\begin{equation}\label{eq_bound_folland_1}
|g_\nu(t)|\le{\rm Const.}\;\!\langle t\rangle^{-(1+\varepsilon)}
\quad\nu\in(0,1),~t\in\R.
\end{equation}
We only show it for the first term in $g_\nu(t)$, namely,
$$
g_{\nu,1}(t):=\int_{\R^d}\underline\d x\,(\F f)(x)
\left\langle\varphi,\eta\big(V(\nu x)^2\big)
\e^{it\int_0^\nu\d s\,(x\cdot V(sx))E^{V^2}(I)}\varphi\right\rangle_{\H_0},
$$
since the second term can be handled in a similar way. Let
$B_{\nu,t}^I(x):=\e^{it\int_0^\nu\d s\,(x\cdot V(sx))E^{V^2}(I)}$. Then, Lemma
\ref{lemma_trotter}(b) and the equality $\varphi=E^{V^2}(I)\varphi$ imply that
$$
\big(\partial_jB_{\nu,t}^I\big)(x)\varphi
=it\nu\;\!V_j(\nu x)B_{\nu,t}^I(x)\varphi.
$$
Thus, $g_{\nu,1}(t)$ can be written for $\nu\in(0,1)$ and $t\in\R\setminus\{0\}$ as
$$
g_{\nu,1}(t)
=\frac1{it\nu}\sum_{j=1}^d\int_{\R^d}\underline\d x\,\big\langle\varphi,
C_{j,\nu}(x)\big(\partial_jB_{\nu,t}^I\big)(x)\varphi\big\rangle_{\H_0},
$$
with
$$
C_{j,\nu}(x):=(\F f)(x)\;\!\eta\big(V(\nu x)^2\big)V_j(\nu x)V(\nu x)^{-2}.
$$
Moreover, one can show as in point (ii) of the proof of Theorem \ref{thm_summation}
that the map $\R^d\ni x\mapsto C_{j,\nu}(x)\in\B(\H_0)$ is twice strongly
differentiable, with strong derivatives satisfying for all $j,k\in\{1,\ldots,d\}$ and
$m\in\N$
$$
\big\|\big(\partial_jC_{j,\nu}\big)(x)\big\|_{\B(\H_0)}
\le{\rm Const.}\;\!\langle\nu\rangle\langle x\rangle^{-m}
$$
and
\begin{equation}\label{eq_bound_C_j}
\big\|\partial_k\big\{\big(\partial_jC_{j,\nu}\big)V_k(\nu\;\!\cdot\;\!)
V(\nu\;\!\cdot\;\!)^{-2}\big\}(x)\big\|_{\B(\H_0)}
\le{\rm Const.}\;\!\langle\nu\rangle^2\langle x\rangle^{-m}.
\end{equation}
So, we can perform two successive integrations by parts with vanishing boundary
contributions to get
\begin{align*}
g_{\nu,1}(t)
&=\frac i{t\nu}\sum_{j=1}^d\int_{\R^d}\underline\d x\,\big\langle\varphi,
\big(\partial_jC_{j,\nu}\big)(x)B_{\nu,t}^I(x)\varphi\big\rangle_{\H_0}\\
&=-\frac 1{t^2\nu^2}\sum_{j,k=1}^d\int_{\R^d}\underline\d x\,\big\langle\varphi,
\partial_k\big\{\big(\partial_jC_{j,\nu}\big)V_k(\nu\;\!\cdot\;\!)
V(\nu\;\!\cdot\;\!)^{-2}\big\}(x)B_{\nu,t}^I(x)\varphi\big\rangle_{\H_0}.
\end{align*}
Combining this with the bound \eqref{eq_bound_C_j}, we obtain for $\nu\in(0,1)$ and
$t\in\R\setminus\{0\}$ that
\begin{equation}\label{eq_bound_g_nu}
|g_{\nu,1}(t)|\le{\rm Const.}\;\!t^{-2}\nu^{-2}\langle\nu\rangle^2.
\end{equation}
Since the function $t\mapsto g_{\nu,1}(t)$ is uniformly bounded in $t\in\R$, the bound
\eqref{eq_bound_g_nu} implies that $g_{\nu,1}$ satisfies \eqref{eq_bound_folland_1}.

We now show that there exists $\varepsilon>0$ such that
\begin{equation}\label{eq_bound_folland_2}
\left|\int_\R\d t\,\e^{-2\pi int}g_\nu(t)\right|
\le{\rm Const.}\;\!\langle n\rangle^{-(1+\varepsilon)}\quad\nu\in(0,1),~n\in\Z.
\end{equation}
The equation \eqref{eq_derivative_g}, together with calculations as above, shows that
we have for $\nu\in(0,1)$ and $t\in\R\setminus\{0\}$ the estimate
\begin{equation}\label{eq_g_nu_derivative}
\big|g_\nu^{(k)}(t)\big|\le{\rm Const.}\;\!t^{-2}\nu^{k-2}\langle\nu\rangle^2.
\end{equation}
Thus, we can perform for $n\in\Z^*$ two successive integrations by parts with
vanishing boundary contributions to obtain the bound
$$
\left|\int_\R\d t\,\e^{-2\pi int}g_\nu(t)\right|
=\left|(2\pi n)^{-2}\int_\R\d t\,\e^{-2\pi int}g^{(2)}_\nu(t)\right|
\le{\rm Const.}\;\!n^{-2}\nu^{k-2}\langle\nu\rangle^2,
$$
which implies \eqref{eq_bound_folland_2}.

(ii) Point (i) shows that we can apply the Poisson summation formula
\cite[Thm.~8.32]{Fol99} to get the equality
$$
\lim_{\nu\searrow0}\sum_{n\in\Z}g_\nu(n)
=\lim_{\nu\searrow0}\int_\R\d t\,g_\nu(t)
+\lim_{\nu\searrow0}\sum_{n\in\Z^*}\int_\R\d t\,\e^{-2\pi int}g_\nu(t).
$$
Furthermore, due to the estimate \eqref{eq_g_nu_derivative}, we can perform two
successive integrations by parts with vanishing boundary contributions in the integral
$\int_\R\d t\,\e^{-2\pi int}g_\nu(t)$ to obtain
\begin{equation}\label{eq_poisson_bis}
\lim_{\nu\searrow0}\sum_{n\in\R}g_\nu(n)
=\lim_{\nu\searrow0}\int_\R\d t\,g_\nu(t)-\lim_{\nu\searrow0}\sum_{n\in\Z^*}
(2\pi n)^{-2}\int_\R\d t\,\e^{-2\pi int}g_\nu^{(2)}(t).
\end{equation}
Now, due to the estimate \eqref{eq_g_nu_derivative} and the convergence of the sum
$\sum_{n\in\Z^*}n^{-2}$, we can apply Lebesgue's dominated convergence theorem in the
last term of \eqref{eq_poisson_bis} to get
\begin{align*}
\lim_{\nu\searrow0}\sum_{n\in\R}g_\nu(n)
&=\lim_{\nu\searrow0}\int_\R\d t\,g_\nu(t)-\sum_{n\in\Z^*}
(2\pi n)^{-2}\int_\R\d t\,\lim_{\nu\searrow0}\e^{-2\pi int}g_\nu^{(2)}(t)
=\lim_{\nu\searrow0}\int_\R\d t\,g_\nu(t),
\end{align*}
which proves the claim.
\end{proof}

For the second technical lemma, we need the following assumption on the $C_0$-group
$\{\e^{ix\cdot Q}\}_{x\in\R^d}:$

\begin{Assumption}\label{ass_polynomial}
The $C_0$-group $\{\e^{ix\cdot Q}\}_{x\in\R^d}$ is of polynomial growth in
$\dom(\langle V\rangle)$, that is, there exists $r>0$ such that
$$
\big\|\e^{ix\cdot Q}\big\|_{\B(\dom(\langle V\rangle),\dom(\langle V\rangle))}
\le{\rm Const.}\;\!\langle x\rangle^r\quad\hbox{for all $x\in\R^d$.}
$$
\end{Assumption}

\begin{Lemma}\label{lemma_taylor}
Let Assumptions \ref{ass_regularity}, \ref{ass_commute} and \ref{ass_polynomial} be
satisfied, and take $\eta\in C_c^\infty\big((0,\infty)\big)$ and $I\subset(0,\infty)$
a compact set. Then, there exists $s>0$ such that
$$
\left\|\tfrac1\nu\left(\eta\big(V(\nu x)^2\big)
\e^{i\frac\mu\nu\int_0^\nu\d s\,(x\cdot V(sx))E^{V^2}(I)}
-\eta\big(V^2\big)\e^{i\mu x\cdot VE^{V^2}(I)}\right)\right\|_{\B(\H_0)}
\le{\rm Const.}\;\!(1+|\mu|)\langle x\rangle^s
$$
for all $\nu\in(0,1)$, $\mu\in\R$ and $x\in\R^d$.
\end{Lemma}

\begin{proof}
For $x\in\R^d$ and $\mu\in\R$, we define the function
$$
g_{x,\mu}:(0,1)\to\B(\H_0),~~\nu\mapsto
\e^{i\frac\mu\nu\int_0^\nu\d s\,(x\cdot V(sx))E^{V^2}(I)}\eta\big(V^2\big),
$$
which is continuous and satisfies in $\B(\H_0)$
\begin{equation}\label{eq_lim_g}
g_{x,\mu}(0)
:=\lim_{\nu\searrow0}g_{x,\mu}(\nu)
=\e^{i\mu x\cdot VE^{V^2}(I)}\eta\big(V^2\big).
\end{equation}
Since $\eta\big(V^2\big)\in C_{\rm u}^1(Q)$ (see \cite[Def.~5.1.1(b)]{ABG96}), we also
have in $\B(\H_0)$
\begin{equation}\label{eq_diff_eta}
\tfrac1\nu\left(\eta\big(V(\nu x)^2\big)-\eta\big(V^2\big)\right) 
=\tfrac1\nu\int_0^1\d t\,\frac{\d}{\d t}\;\!\eta\big(V(t\nu x)^2\big) 
=i\sum_{j=1}^d x_j\int_0^1\d t\,\e^{-it\nu x\cdot Q}\big[\eta\big(V^2\big),Q_j\big]
\e^{it\nu x\cdot Q}.
\end{equation}
Therefore, by combining the identities \eqref{eq_lim_g} and \eqref{eq_diff_eta}, we
obtain
\begin{align}
&\tfrac1\nu\left(\eta\big(V(\nu x)^2\big)
\e^{i\frac\mu\nu\int_0^\nu\d s\,(x\cdot V(sx))E^{V^2}(I)}
-\eta\big(V^2\big)\e^{i\mu x\cdot V E^{V^2}(I)}\right)\nonumber\\
&=\tfrac1\nu\left(\eta\big(V(\nu x)^2\big)-\eta\big(V^2\big)\right)
\e^{i\frac\mu\nu\int_0^\nu\d s\,(x\cdot V(sx))E^{V^2}(I)}
+\tfrac1\nu\big(g_{x,\mu}(\nu)-g_{x,\mu}(0)\big)\nonumber\\
&=i\sum_{j=1}^dx_j\int_0^1\d t\,\e^{-it\nu x\cdot Q}\big[\eta\big(V^2\big),Q_j\big]
\e^{it\nu x\cdot Q}\e^{i\frac\mu\nu\int_0^\nu\d s\,(x\cdot V(sx))E^{V^2}(I)}
+\tfrac1\nu\big(g_{x,\mu}(\nu)-g_{x,\mu}(0)\big).\label{eq_big_diff}
\end{align}
Now, since
$
\frac1\nu\int_0^\nu\d s\,\big(x\cdot V(sx)\big)E^{V^2}(I)
=\int_0^1\d t\,\big(x\cdot V(t\nu x)\big)E^{V^2}(I)
$,
we have for $\varepsilon\in\R\setminus\{0\}$ small enough
\begin{align*}
&g_{x,\mu}(\nu+\varepsilon)-g_{x,\mu}(\nu)\\
&=\left(\e^{i\mu\int_0^1\d t\,(x\cdot V(t(\nu+\varepsilon)x))E^{V^2}(I)}
-\e^{i\mu\int_0^1\d t\,(x\cdot V(t\nu x))E^{V^2}(I)}\right)\eta\big(V^2\big)\\
&=\e^{i\mu\int_0^1\d t\,(x\cdot V(t\nu x))E^{V^2}(I)}
\left(\e^{i\mu\int_0^1\d t\,x\cdot(V(t(\nu+\varepsilon)x)-V(t\nu x))
E^{V^2}(I)}-1\right)\eta\big(V^2\big)\\
&=\e^{i\mu\int_0^1\d t\,(x\cdot V(t\nu x))E^{V^2}(I)}
\left(\e^{i\mu\int_0^1\d t\int_0^1\d s\,t\varepsilon\sum_{j,k=1}^dx_jx_k
V_{jk}'(t(\nu+s\varepsilon)x)E^{V^2}(I)}-1\right)\eta\big(V^2\big),
\end{align*}
and multiplying by $\varepsilon^{-1}$ and taking the limit $\varepsilon\to0$ in
$\B(\H_0)$ we obtain
\begin{equation}\label{eq_der_g}
g_{x,\mu}'(\nu)
=i\mu\e^{i\mu\int_0^1\d t\,(x\cdot V(t\nu x))E^{V^2}(I)} 
\int_0^1\d t\,t\sum_{j,k=1}^dx_jx_kV_{jk}'(t\nu x)E^{V^2}(I)\eta\;\!\big(V^2\big).
\end{equation}
This, together with \eqref{eq_big_diff} and the mean value theorem, implies that
\begin{align}
&\left\|\tfrac1\nu\left(\eta\big(V(\nu x)^2\big)
\e^{i\frac\mu\nu\int_0^\nu\d s\,(x\cdot V(sx))E^{V^2}(I)}
-\eta\big(V^2\big)\e^{i\mu x\cdot V E^{V^2}(I)}\right)\right\|_{\B(\H_0)}\nonumber\\
&\le{\rm Const.}\;\!|x|
+\sup_{\xi\in [0,1]}\big\|g_{x,\mu}'(\xi \nu)\big\|_{\B(\H_0)}\nonumber\\
&\le{\rm Const.}\;\!|x|+{\rm Const.}\;\!|x|^2\;\!|\mu|\sup_{\xi\in[0,1]}
\sum_{j,k=1}^d\big\|V_{jk}'(\xi\nu x)E^{V^2}(I)\;\!\eta\big(V^2\big)\big\|_{\B(\H_0)}.
\label{eq_big_bound}
\end{align}
Since
$$
V_{jk}'(\xi\nu x)E^{V^2}(I)\;\!\eta\big(V^2\big)
=\e^{-i\xi\nu x \cdot Q}V_{jk}'\e^{i\xi\nu x\cdot Q}E^{V^2}(I)\;\!\eta\big(V^2\big),
$$
with $E^{V^2}(I)\;\!\eta\big(V^2\big)\in\B\big(\H_0,\dom(\langle V\rangle)\big)$ and
$V_{jk}'\in\B\big(\dom(\langle V\rangle),\H_0)$, it follows from Assumption
\ref{ass_polynomial} that there exists $r>0$ such that 
$$
\big\|V_{jk}'(\xi\nu x)E^{V^2}(I)\;\!\eta\big(V^2\big)\big\|_{\B(\H_0)}
\le{\rm Const.}\;\!\langle\xi\nu x\rangle^r.
$$
Therefore, we obtain from \eqref{eq_big_bound} that
$$
\left\|\tfrac1\nu\left(\eta\big(V(\nu x)^2\big)
\e^{i\frac\mu\nu\int_0^\nu\d s\,(x\cdot V(sx))E^{V^2}(I)}
-\eta\big(V^2\big)\e^{i\mu x\cdot V E^{V^2}(I)}\right)\right\|_{\B(\H_0)}
\le{\rm Const.}\;\!(1+|\mu|)\langle x\rangle^{r+2},
$$
and thus the claim follows with $s:=r+2$.
\end{proof}

For the next proposition, we need to define for $\nu>0$ and for $f\in\SS(\R^d)$ equal
to $1$ on a neighbourhood of $0\in\R^d$ the function
$$
F_{\nu,f}:\R^d\setminus\{0\}\to\C,~~x\mapsto\sum_{n\in\Z}f(\nu nx).
$$
The function $F_f$ is well-defined because $f\in\SS(\R^d)$, and if $f$ is radial or
$d=1$, then $F_{\nu,f}(x)$ depends only on the squared norm $x^2$ of
$x\in\R^d\setminus\{0\}$.

\begin{Proposition}\label{prop_equal_sojourn}
Let Assumptions \ref{ass_regularity}, \ref{ass_commute}, \ref{ass_wave} and
\ref{ass_polynomial} be satisfied. Let $f\in\SS(\R^d)$ be non-negative, even and equal
to $1$ on a neighbourhood of $0\in\R^d$. Let $\varphi\in\H_0^-\cap\dom_2$ satisfy
$S\varphi\in\dom_2$. Finally, suppose that $V^2$ and $S$ strongly commute and that
\begin{equation}\label{eq_commute_S}
\big[\eta\big(V^2\big)F_{\nu,f}(V),S\big]=0
\quad\hbox{for all $\eta\in C^\infty_{\rm c}\big((0,\infty)\big)$ and $\nu>0$.}
\end{equation}
Then, one has
$$
\lim_{\nu\searrow0}\big(T_{1/\nu}^0(S\varphi)-T_{1/\nu}^0(\varphi)\big)=0.
$$
\end{Proposition}

The operator $\eta\big(V^2\big)F_{\nu,f}(V)$ in \eqref{eq_commute_S} is well-defined
and bounded because $\eta\in C^\infty_{\rm c}\big((0,\infty)\big)$ and
$F_{\nu,f}:\R^d\setminus\{0\}\to\C$. Furthermore, if $f$ is radial or $d=1$, then
$\eta\big(V^2\big)F_{\nu,f}(V)$ can be written as a function of $V^2$ and the
condition \eqref{eq_commute_S} automatically follows from the strong commuation of
$V^2$ and $S$.

\begin{proof}
Let $\varphi\in\H_0^-\cap\dom_2$ satisfy $S\varphi\in\dom_2$. Then, there exist a real
function $\eta\in C^\infty_{\rm c}\big((0,\infty)\big)$ and a compact set
$I\subset(0,\infty)$ such that $\varphi=\eta\big(V^2\big)\varphi=E^{V^2}(I)\varphi$.
This, together with \eqref{eq_commute_S}, Lemma \ref{lemma_trotter}(a) and the strong
commutation of $V^2$ and $S$, implies
\begin{align}
&T_{1/\nu}^0(S\varphi)-T_{1/\nu}^0(\varphi)\nonumber\\
&=\sum_{n\in\Z}\left\langle\varphi,S^*\big[\eta\big(V^2\big)
U_0^{-n}f(\nu Q)U_0^nE^{V^2}(I),S\big]\varphi\right\rangle_{\H_0}
-\left\langle\varphi,S^*\big[\eta\big(V^2\big)F_{\nu,f}(V)E^{V^2}(I),S\big]
\varphi\right\rangle_{\H_0}\nonumber\\
&=\sum_{n\in\Z}\left\langle\varphi,
S^*\big[\eta\big(V^2\big)U_0^{-n}f(\nu Q)U_0^nE^{V^2}(I)
-\eta\big(V^2\big)f(\nu n\;\!V)E^{V^2}(I),S\big]\varphi\right\rangle_{\H_0}\nonumber\\
&=\sum_{n\in\Z}\int_{\R^d}\underline\d x\,(\F f)(x)\left\langle\varphi,
S^*\Big[\eta\big(V^2\big)\e^{i\nu x\cdot Q}
\e^{in\int_0^\nu\d s\,(x\cdot V(sx))E^{V^2}(I)}
-\eta\big(V^2\big)\e^{i\nu nx\cdot VE^{V^2}(I)},S\Big]
\varphi\right\rangle_{\H_0}\nonumber\\
&=\sum_{n\in\Z}\int_{\R^d}\underline\d x\,(\F f)(x)\left\langle\varphi,
S^*\Big[\e^{i\nu x\cdot Q}\eta\big(V(\nu x)^2\big)
\e^{in\int_0^\nu\d s\,(x\cdot V(sx))E^{V^2}(I)}
-\eta\big(V^2\big)\e^{i\nu nx\cdot VE^{V^2}(I)},S\Big]
\varphi\right\rangle_{\H_0}\nonumber\\
&=\sum_{n\in\Z}\int_{\R^d}\underline\d x\,(\F f)(x)\left\langle\varphi,
S^*\Big[\big(\e^{i\nu x\cdot Q}-1\big)\eta\big(V(\nu x)^2\big)
\e^{in\int_0^\nu\d s\,(x\cdot V(sx))E^{V^2}(I)},S\Big]
\varphi\right\rangle_{\H_0}\nonumber\\
&\quad+\sum_{n\in\Z}\int_{\R^d}\underline\d x\,(\F f)(x)\left\langle\varphi,
S^*\Big[\eta\big(V(\nu x)^2\big)\e^{in\int_0^\nu\d s\,(x\cdot V(sx))E^{V^2}(I)}
-\eta\big(V^2\big)\e^{i\nu nx\cdot VE^{V^2}(I)},S\Big]\varphi\right\rangle_{\H_0}.
\label{eq_two_terms}
\end{align}
So, to prove the claim, it is sufficient to show that the two terms in
\eqref{eq_two_terms} are equal to zero in the limit $\nu\searrow0$. This is done in
points (i) and (ii) below.

(i) For the first term in \eqref{eq_two_terms}, we can adapt the proof of Theorem
\ref{thm_summation} to obtain the equalities
\begin{align*}
& \lim_{\nu\searrow0}\sum_{n\in\Z}\int_{\R^d} \underline\d x\,(\F f)(x)
\left\langle\varphi,S^*\Big[\big(\e^{i\nu x\cdot Q}-1\big)\eta\big(V(\nu x)^2\big)
\e^{in\int_0^\nu\d s\,(x\cdot V(sx))E^{V^2}(I)},S\Big]\varphi\right\rangle_{\H_0}\\
&=\lim_{\nu\searrow0}\int_\R\d\mu\int_{\R^d}\underline\d x\,(\F f)(x)
\left\langle\varphi,S^*\Big[\tfrac1\nu\big(\e^{i\nu x\cdot Q}-1\big)
\eta\big(V(\nu x)^2\big)\e^{i\frac\mu\nu\int_0^\nu\d s\,
(x\cdot V(sx))E^{V^2}(I)},S\Big]\varphi\right\rangle_{\H_0}\\
&=i\int_\R\d\mu\int_{\R^d}\underline\d x\,(\F f)(x)
\left(\left\langle(x\cdot Q)S\varphi,\e^{i\mu x\cdot V}S\varphi\right\rangle_{\H_0} 
-\left\langle(x\cdot Q)\varphi,\e^{i\mu x\cdot V}\varphi\right\rangle_{\H_0}\right),
\end{align*}
and then the change of variables $\mu':=-\mu$ and $x':=-x$ together with the parity of
$\F f$ implies that the last expression is equal to zero.

(ii) For the second term in \eqref{eq_two_terms}, it is sufficient to prove for
$\psi\in\dom_2$ satisfying $\psi=\eta\big(V^2\big)\psi=E^{V^2}(I)\psi$ that
$$
\lim_{\nu\searrow0}\sum_{n\in\Z}\int_{\R^d}\underline\d x\,(\F f)(x)
\left\langle\psi,\left(\eta\big(V(\nu x)^2\big)
\e^{in\int_0^\nu\d s\,(x\cdot V(sx))E^{V^2}(I)}-\eta\big(V^2\big)
\e^{i\nu nx\cdot V E^{V^2}(I)}\right)\psi\right\rangle_{\H_0}=0.
$$
Using Lemma \ref{lemma_poisson} and the change of variables $\mu:=\nu t$ we obtain
that this is equivalent to
\begin{equation}\label{eq_second_term}
\lim_{\nu\searrow0}\int_\R\d\mu\int_{\R^d}\underline\d x\,(\F f)(x)
\left\langle\psi,\tfrac1\nu\left(\eta\big(V(\nu x)^2\big)
\e^{i\frac\mu\nu\int_0^\nu\d s\,(x\cdot V(sx))E^{V^2}(I)}
-\eta\big(V^2\big)\e^{i\mu x\cdot VE^{V^2}(I)}\right)\psi\right\rangle_{\H_0}=0.
\end{equation}
Now, under the assumption that we can exchange in \eqref{eq_second_term} the limit
$\lim_{\nu\searrow0}$ and the integrals over $\mu$ and $x$, then it follows by
\eqref{eq_big_diff}-\eqref{eq_der_g} that
\begin{align*}
&\lim_{\nu\searrow0}\int_\R\d\mu\int_{\R^d}\underline\d x\,(\F f)(x)
\left\langle\psi,\tfrac1\nu\left(\eta\big(V(\nu x)^2\big)
\e^{i\frac\mu\nu\int_0^\nu\d s\,(x\cdot V(sx))E^{V^2}(I)}
-\eta\big(V^2\big)\e^{i\mu x\cdot VE^{V^2}(I)}\right)\psi\right\rangle_{\H_0}\\
&=i\int_\R\d\mu\int_{\R^d}\underline\d x\,(\F f)(x)
\left\langle\psi,\Bigg(\big[\eta\big(V^2\big),x\cdot Q\big]\e^{i\mu x\cdot V}
+\frac\mu2\e^{i\mu x\cdot V}\sum_{j,k=1}^dx_jx_kV_{jk}'\;\!\eta\big(V^2\big)\Bigg)
\psi\right\rangle_{\H_0},
\end{align*}
and then the change of variables $\mu':=-\mu$ and $x':=-x$ together with the parity of
$\F f$ implies that the last expression is equal to zero. Thus, it only remains to show
that we can exchange in \eqref{eq_second_term} the limit $\lim_{\nu\searrow0}$ and the
integrals over $\mu$ and $x$ by applying Lebesgue's dominated convergence theorem.

Define for $\nu\in(0,1)$ and $\mu\in\R$
$$
L(\nu,\mu):=\int_{\R^d}\underline\d x\,(\F f)(x)
\left\langle\psi,\tfrac1\nu\Big(\eta\big(V(\nu x)^2\big)
\e^{i\frac\mu\nu\int_0^\nu\d s\,(x\cdot V(sx))E^{V^2}(I)}
-\eta\big(V^2\big)\e^{i\mu x\cdot VE^{V^2}(I)}\Big)\psi\right\rangle_{\H_0}.
$$
Then Lemma \ref{lemma_taylor} and the rapid decay of $\F f$ imply that
$$
|L(\nu,\mu)|\le{\rm Const.}\;\!(1+|\mu|),
$$
with a constant independent of $\nu$. Thus, $|L(\nu,\mu)|$ is uniformly bounded 
in $\nu\in(0,1)$ by a function in $\lone\big([-1,1],\d\mu\big)$.

For the case $|\mu|>1$, set
$$
A_\mu^I(x):=\e^{i\mu x\cdot VE^{V^2}(I)}
\quad\hbox{and}\quad
B_{\nu,\mu}^I(x):=\e^{i\frac\mu\nu\int_0^{\nu}\d s\,(x\cdot V(sx))E^{V^2}(I)} .
$$ 
Then, we have
$$
\big(\partial_jA_\mu^I\big)(x)\psi
=i\mu V_jA_\mu^I(x)\psi,
$$
and Lemma \ref{lemma_trotter}(b) implies that
$$
\big(\partial_jB_{\nu,\mu}^I\big)(x)\psi
=i\mu V_j(\nu x)B_{\nu,\mu}^I(x)\psi.
$$
Therefore, with the notations $C_j:=\eta\big(V^2\big)V_j\;\!V^{-2}$ and
$G_x:=\e^{-ix\cdot Q}$ we can rewrite $L(\nu,\mu)$ as 
$$
L(\nu,\mu):=\frac1{i\mu}\sum_{j=1}^d\int_{\R^d}\underline\d x\,(\F f)(x) 
\left\langle\psi,\tfrac1\nu\Big(G_{\nu x}C_j\;\!G_{\nu x}^*
\big(\partial_jB_{\nu,\mu}^I\big)(x)
-C_j\big(\partial_jA_\mu^I\big)(x)\Big)\psi\right\rangle_{\H_0}.
$$
We shall now use repeatedly the following result: Let $h \in \SS(\R^d)$ and let
$Y:=(Y_1,\ldots,Y_d)$ a family of self-adjoint strongly mutually commuting operators
in $\H_0$. If $Y_1,\ldots,Y_d$ are of class $C^2(Q)$, then $h(Y)\in C^2(Q)$ and
$\big[[h(Y),Q_j],Q_k\big]\in\B(\H_0)$ for all $j,k\in\{1,\ldots,d\}$. Such a result
has been proved in \cite[Prop,~5.1]{RT12_0} in a greater generality. Here, the
operator $C_j$ is of the type $h(Y)$ since $V_1,\ldots,V_d$ are self-adjoint strongly
mutually commuting operators of class $C^2(Q)$. So, we can perform an integration by
parts (with vanishing  boundary contributions) with respect to the variable $x_j$ to
obtain
\begin{align*}
L(\nu,\mu)
&:=-\frac1{i\mu}\sum_{j=1}^d\int_{\R^d}\underline\d x\,\big(\partial_j(\F f)\big)(x)
\left\langle\psi,\tfrac1\nu\big(G_{\nu x}C_jG_{\nu x}^*B_{\nu,\mu}^I(x)
-C_jA_\mu^I(x)\big)\psi\right\rangle_{\H_0}\\
&\quad-\frac1\mu\sum_{j=1}^d\int_{\R^d}\underline\d x\,(\F f)(x)\left\langle\psi,
G_{\nu x}\big[C_j,Q_j\big]G_{\nu x}^*B_{\nu,\mu}^I(x)\psi\right\rangle_{\H_0}.
\end{align*}
Now, the scalar product in the first term can be written as
$$
\tfrac1{i\mu}\left\langle\psi,\tfrac1\nu\Big(G_{\nu x}EG_{\nu x}^*
\big(\partial_jB_{\nu,\mu}^I \big)(x)
-E\big(\partial_jA_\mu^I\big)(x)\Big)\psi\right\rangle_{\H_0}
\quad\hbox{with}\quad
E:=\eta\big(V^2\big)V^{-2}\in\B(\H_0).
$$
Thus, a second integration by parts leads to
\begin{align}
L(\nu,\mu)
&:=-\frac1{\mu^2}\sum_{j=1}^d\int_{\R^d}\underline\d x\,
\big(\partial_j^2(\F f)\big)(x)\left\langle\psi,\tfrac1\nu
\big(G_{\nu x}EG_{\nu x}^*B_{\nu,\mu}^I(x)-EA_\mu^I(x)\big)
\psi\right\rangle_{\H_0}\nonumber\\
&\quad-\frac i{\mu^2}\sum_{j=1}^d\int_{\R^d}\underline\d x\,
\big(\partial_j(\F f)\big)(x)\left\langle\psi,G_{\nu x}\big[E,Q_j\big]G_{\nu x}^*
B_{\nu,\mu}^I(x)\psi\right\rangle_{\H_0}\nonumber\\
&\quad-\frac1\mu\sum_{j=1}^d\int_{\R^d}\underline\d x\,(\F f)(x)
\left\langle\psi,G_{\nu x}\big[C_j,Q_j\big]G_{\nu x}^*B_{\nu,\mu}^I(x)
\psi\right\rangle_{\H_0}.\label{eq_3_terms}
\end{align}
Then, by performing a third integration by parts, we obtain that the first term in
\eqref{eq_3_terms} is equal to
\begin{align*}
&\frac i{\mu^3}\sum_{j,k=1}^d\int_{\R^d}\underline\d x\,
\big(\partial_j^2(\F f)\big)(x)\left\langle\psi,\tfrac1\nu
\Big(G_{\nu x}M_k G_{\nu x}^*\big(\partial_kB_{\nu,\mu}^I\big)(x)
-M_k\big(\partial_kA_\mu^I\big)(x)\Big)\psi\right\rangle_{\H_0} \\
&=-\frac i{\mu^3}\sum_{j,k=1}^d\int_{\R^d}\underline\d x\,
\big(\partial_k\partial_j^2(\F f)\big)(x)
\left\langle\psi,\tfrac1\nu\Big(G_{\nu x}M_kG_{\nu x}^*B_{\nu,\mu}^I(x)
-M_kA_\mu^I(x)\Big)\psi\right\rangle_{\H_0}\\
&\quad+\frac1{\mu^3}\sum_{j,k=1}^d\int_{\R^d}\underline\d x\,
\big(\partial_j^2(\F f)\big)(x)\left\langle\psi,G_{\nu x}\big[M_k,Q_k\big]G_{\nu x}^*
B_{\nu,\mu}^I(x)\psi\right\rangle_{\H_0}
\end{align*}
with $M_k:=\eta\big(V^2\big)V_kV^{-4}\in\B(\H_0)$. Furthermore, by mimicking the proof
of Lemma \ref{lemma_taylor} with $\eta\big(V^2\big)$ replaced by $M_k$, we obtain that
there exists $s>0$ such that
$$
\Big|\tfrac1\nu\Big(G_{\nu x}M_k G_{\nu x}^*B_{\nu,\mu}^I(x)
-M_kA_\mu^I(x)\Big)\Big|_{\B(\H_0)}\le{\rm Const.}\;\!(1+|\mu|)\langle x\rangle^s
$$
for all $\nu\in(0,1)$, $\mu\in\R$ and $x\in\R^d$. Thus, the first and the second terms
in \eqref{eq_3_terms} can be bounded uniformly in $\nu\in(0,1)$ by a function in
$\lone(\R\setminus[-1,1],\d\mu)$. For  the third term in \eqref{eq_3_terms}, a direct
calculation shows that it can be written as
$$
-\frac i{\mu^2}\sum_{j,k=1}^d\int_{\R^d}\underline\d x\,(\F f)(x)
\left\langle G_{\nu x}^*\psi,[C_j,Q_j]G_{\nu x}^*C_kG_{\nu x}
\big(\partial_kB_{\nu,-\mu}^I\big)(-x)G_{\nu x}^*\psi\right\rangle_{\H_0}.
$$
So, by doing once more an integration by parts with respect to the variable $x_k$, we
also obtain that this term can be bounded uniformly in $\nu\in(0,1)$ by a function in
$\lone(\R\setminus[-1,1],\d\mu)$.

These last estimates together with the previous estimate for $\mu\in[-1,1]$ shows that
$|L(\nu,\mu)|$ is bounded uniformly in $\nu\in(0,1)$ by a function in
$\lone(\R,\d\mu)$. Therefore, we can exchange the limit $\lim_{\nu\searrow0}$ and the
integration over $\mu$ in \eqref{eq_second_term}. Due to Lemma \ref{lemma_taylor}, we
can also exchange the limit $\lim_{\nu\searrow0}$ and the integration over $x$ in 
\eqref{eq_second_term}.
\end{proof}

The existence of the non-symmetrised time delay is now a direct consequence of
Theorems \ref{thm_sym} and Proposition \ref{prop_equal_sojourn}:

\begin{Theorem}[Non-symmetrised time delay]\label{thm_non_sym}
Let Assumptions \ref{ass_regularity}, \ref{ass_commute}, \ref{ass_wave} and
\ref{ass_polynomial} be satisfied. Let $f\in\SS(\R^d)$ be non-negative, even and equal
to $1$ on a neighbourhood of $0\in\R^d$. For each $n\in\Z$, let $L_n:\H\to\H_0$
satisfy $L_nU^nW_\pm E^D(I)\in\B(\H,\H_0)$ for all compact sets $I\subset\R$. Let
$\varphi\in\H_0^-\cap\dom_2$ satisfy $S\varphi\in\dom_2$ and \eqref{eq_l_one}.
Finally, suppose that $V^2$ and $S$ strongly commute and that
$$
\big[\eta\big(V^2\big)F_{\nu,f}(V),S\big]=0
\quad\hbox{for all $\eta\in C^\infty_{\rm c}\big((0,\infty)\big)$ and $\nu>0$.}
$$
Then, one has
$$
\lim_{r\to\infty}\tau_r^{\rm nsym}(\varphi)
=\lim_{r\to\infty}\tau_r^{\rm sym}(\varphi)
=\big\langle\varphi,S^*[T_f,S]\varphi\big\rangle_{\H_0},
$$
with $T_f$ defined in \eqref{eq_T_f}.
\end{Theorem}


\def\cprime{$'$} \def\polhk#1{\setbox0=\hbox{#1}{\ooalign{\hidewidth
  \lower1.5ex\hbox{`}\hidewidth\crcr\unhbox0}}}
  \def\polhk#1{\setbox0=\hbox{#1}{\ooalign{\hidewidth
  \lower1.5ex\hbox{`}\hidewidth\crcr\unhbox0}}}
  \def\polhk#1{\setbox0=\hbox{#1}{\ooalign{\hidewidth
  \lower1.5ex\hbox{`}\hidewidth\crcr\unhbox0}}} \def\cprime{$'$}
  \def\cprime{$'$} \def\polhk#1{\setbox0=\hbox{#1}{\ooalign{\hidewidth
  \lower1.5ex\hbox{`}\hidewidth\crcr\unhbox0}}}
  \def\polhk#1{\setbox0=\hbox{#1}{\ooalign{\hidewidth
  \lower1.5ex\hbox{`}\hidewidth\crcr\unhbox0}}} \def\cprime{$'$}
  \def\cprime{$'$} \def\cprime{$'$}


\end{document}